\documentclass[journal]{IEEEtran}
\usepackage{cite}
\usepackage{amsmath,amssymb,amsfonts}
\usepackage{algorithmic}
\usepackage{graphicx}
\usepackage{algorithm,algorithmic}
\usepackage{hyperref}
\hypersetup{hidelinks=true}
\usepackage{textcomp}
\usepackage{xcolor}

\usepackage{amsmath,amssymb,mathtools,bm,upgreek}
\newtheorem{definition}{Definition}
\newtheorem{theorem}{Theorem}
\newtheorem{lemma}{Lemma}
\newtheorem{proposition}{Proposition}
\newtheorem{corollary}{Corollary}
\newtheorem{assumption}{Assumption}
\newtheorem{remark}{Remark}

\newenvironment{proof}{\begin{IEEEproof}}{\end{IEEEproof}}
\newcommand{\vsgn}[2]{\left|\mkern-2mu\left\lceil #1 \right\rfloor\mkern-2mu\right|^{#2}}
\newcommand{\R}{\mathbb{R}}

\def\BibTeX{{\rm B\kern-.05em{\sc i\kern-.025em b}\kern-.08em
    T\kern-.1667em\lower.7ex\hbox{E}\kern-.125emX}}

\begin{document}
\title{Abstract homogeneous chains:\\ a Lyapunov framework for high-order\\ sliding modes in multi-agent systems}
\author{Rodrigo Aldana-López
\thanks{\textcolor{red}{This work has been submitted to the IEEE for possible publication. Copyright may be transferred without notice, after which this version may no longer be accessible.}}
\thanks{This work was supported via projects PID2021-124137OB-I00 and PID2024-159279OB-I00 funded by MICIU/AEI/10.13039/501100011033 and by ERDF/EU, via project REMAIN S1/1.1/E0111 (Interreg Sudoe Programme, ERDF), by the Gobierno de Aragón under Project DGA T45-23R, and by Spanish grant FPU20/03134. Grant reference BG24/00121 funded by MICIU/AEI/10.13039/501100011033.}
\thanks{Rodrigo Aldana-López is with the Departamento de Informática e Ingeniería de Sistemas (DIIS) and Instituto de Investigación en Ingeniería de Aragón (I3A), Universidad de Zaragoza, 50018 Zaragoza, Spain (e-mail: raldana@unizar.es).}}

\maketitle
\begin{abstract}
This work develops a Lyapunov framework for a broad class of arbitrary-order sliding-mode algorithms in multi-agent systems. We introduce abstract homogeneous chains, a class of nonlinear error systems characterized by common convexity and homogeneity properties. For this class, we establish global finite-time stability for arbitrary order, construct a homogeneous Lyapunov function, and derive a recursive optimization-based gain-proposal procedure. The framework addresses several gaps in existing dynamic average consensus and distributed differentiation results: it provides a recursive numerical optimization-based gain-proposal procedure for EDCHO at arbitrary order, extends REDCHO convergence from local to global, and provides arbitrary-order numerical gain-proposal rules for leader--follower distributed differentiation, previously available only at first order. It also provides a new arbitrary-order observer for multi-leader affine formation tracking with global finite-time convergence. 
\end{abstract}

\begin{IEEEkeywords}
High-order sliding modes, multi-agent systems, finite-time stability, homogeneity, Lyapunov methods, distributed estimation, dynamic average consensus, formation control.
\end{IEEEkeywords}

\section{Introduction}
\label{sec:introduction}

\IEEEPARstart{H}{igh-order} sliding-mode (HOSM) algorithms have become a well-established tool for robust finite-time estimation and control. Initially motivated by the possibility of retaining the robustness properties of sliding modes while alleviating the chattering associated with first-order discontinuous feedback, HOSM techniques also enabled robust exact differentiation of arbitrary order under bounded higher-order variations \cite{Levant1998,levant2003}. Their theoretical analysis has evolved substantially over the years. Early results relied mainly on geometric and homogeneity arguments to establish finite-time convergence and the existence of stabilizing gains \cite{levant2003,levant2005}. Lyapunov constructions were subsequently developed for the so-called super-twisting algorithm \cite{moreno2012}, leading to explicit feasible gain conditions and settling-time estimates \cite{seeber2017,seeber2018}, and eventually to Lyapunov-based analyses and tuning rules for arbitrary-order exact differentiators \cite{cruz2018,seeber2023}. This progression has produced a mature theory in which finite-time convergence, Lyapunov analysis, and gain design can be addressed systematically for classical HOSM algorithms.

The use of sliding modes in multi-agent systems has followed a similar progression at the algorithmic level, but not yet at the level of the supporting theory. For example, first-order sliding modes were introduced in dynamic average consensus to achieve exact tracking under bounded persistent variations, overcoming limitations of linear consensus mechanisms \cite{freeman2019}. HOSM then appeared through super-twisting-type distributed differentiators and were later extended to arbitrary order. In particular, the Exact Dynamic Consensus of High Order (EDCHO) protocol and its robust extension REDCHO provide leaderless distributed differentiation of arbitrary order \cite{edcho,redcho}, while related HOSM constructions have been proposed for leader--follower distributed differentiation \cite{elf}. These developments show that HOSM mechanisms can be successfully embedded in networked systems. However, the corresponding Lyapunov framework has not progressed at the same pace. For EDCHO, the available gain-design procedure leaves one gain to be chosen sufficiently large, without providing an explicit characterization of the required threshold for arbitrary order. For REDCHO, the available convergence guarantee is local. For the leader--follower case, Lyapunov-based gain-proposal is available only for first-order differentiation. 

The difficulty is structural. In classical differentiators, the nonlinear feedback acts directly on a single scalar or vector estimation error. In multi-agent systems, these nonlinearities are intertwined with the communication topology through incidence matrices, Laplacians, leader-pinning matrices, or other graph-dependent operators. Consequently, existing generalized or multivariable super-twisting frameworks do not directly cover the couplings arising in networked systems. Constructive multivariable generalizations have been developed in \cite{lopez2019generalised,moreno2021multivariable,geromel2026}, but they rely on structural relations between the continuous and discontinuous feedback channels that are not satisfied by general graph-induced nonlinearities. Likewise, approaches based on transformations that reveal a suitably scaled linear structure \cite{hernan2019,seeber2021} do not extend directly when the nonlinear feedback is coupled through a communication graph.

A different viewpoint was recently proposed in \cite{aldana2026b} for the super-twisting case. Rather than preserving the particular algebraic form of the classical algorithm, that work identified structural properties to define an abstract super-twisting system encompassing both the classical super-twisting algorithm and a leaderless distributed differentiator. By combining homogeneity with convex analysis, this abstraction enabled a common Lyapunov construction, global finite-time stability, and numerical gain conditions for the latter. However, the framework is restricted to the second-order super-twisting structure, and its connection with other multi-agent problems was not explored.

The present work generalizes this viewpoint to a broad class of arbitrary-order multi-agent HOSM error dynamics characterized by common homogeneity, convexity, and graph-induced properties. We formalize this class as an \emph{Abstract Homogeneous Chain} (AHC), which provides a unified Lyapunov and gain-proposal framework. The main contributions are summarized as follows.

\begin{enumerate}
\item We introduce AHCs and establish global finite-time stability for this class at arbitrary order. The proof develops a homogeneous Lyapunov construction combining the scaling properties underlying HOSM algorithms with tools from convex analysis. 

\item The Lyapunov analysis yields a recursive gain-proposal procedure in which all gain thresholds are characterized by finite-dimensional numerical optimization problems. Applied to existing distributed HOSM algorithms, this addresses several gaps in the literature: it provides an optimization-based arbitrary-order gain-proposal procedure for EDCHO, extends the REDCHO convergence result from local to global at arbitrary order, and provides arbitrary-order numerical gain-proposal rules for the leader--follower distributed differentiator, for which parameter tuning was available only at first order.

\item The framework is also shown to have design value beyond the analysis of existing algorithms. We use the AHC structure to construct a new arbitrary-order protocol for multi-leader affine formation with global finite-time convergence.
\end{enumerate}

\section{Preliminaries}

\subsection{Notation}
Lowercase bold letters denote vectors and uppercase bold letters denote matrices. Let \(\mathcal X\subseteq\R^n\) be a linear subspace endowed with the Euclidean inner product \(\mathbf x^\top\mathbf y\) and induced vector and matrix norm \(\|\bullet\|\). Let \(\mathbf0_n\), \(\mathbf1_n\), and \(\mathbf I_n\) denote the zero vector, all-ones vector, and identity matrix, respectively, with the subscript omitted when the dimension is clear. For compatible vectors, define \(\operatorname{col}(\mathbf x_1,\ldots,\mathbf x_q):=[\mathbf x_1^\top\ \cdots\ \mathbf x_q^\top]^\top\), and let \(\operatorname{diag}(\mathbf A_1,\ldots,\mathbf A_q)\) denote the corresponding diagonal or block-diagonal matrix. For a finite index set \(\mathcal I\subset\mathbb N\), \(\operatorname{col}_{i\in\mathcal I}(\mathbf x_i)\) denotes the column stacking of the vectors \(\mathbf x_i\), with the indices taken in increasing order. For a linear map \(\mathbf{Q}:\mathcal X\to\mathcal X\), with \(\mathbf Q^{-1}\) denoting its inverse on \(\mathcal X\) whenever it exists, a scalar \(a\in\R\), and sets \(\mathcal A,\mathcal B\subseteq\mathcal X\), write \(\mathbf{Q}\mathcal A:=\{\mathbf{Q}\mathbf a:\mathbf a\in\mathcal A\}\) and \(a\mathcal A-\mathcal B:=\{a\mathbf a-\mathbf b:\mathbf a\in\mathcal A,\ \mathbf b\in\mathcal B\}\). For \(a\in\R\) and \(\alpha>0\), define
\(
\lceil a\rfloor^\alpha:=|a|^\alpha\operatorname{sign}(a).
\)
The function \(\operatorname{sign}\) has its usual single-valued meaning, with \(\operatorname{sign}(0)=0\), except in differential inclusions, where the same notation is used with \(\operatorname{sign}(0)=[-1,1]\). For vector arguments, \(\lceil\cdot\rfloor^\alpha\) and \(\operatorname{sign}(\cdot)\) are applied componentwise. For \(\mathbf x\in\R^d\) and \(\alpha\geq0\), define 
\[
\vsgn{\mathbf x}{\alpha}
:=
\begin{cases}
\|\mathbf x\|^{\alpha-1}\mathbf x, & \mathbf x\neq\mathbf0,\\
\mathbf0, & \mathbf x=\mathbf0,\ \alpha>0,\\
\{\mathbf v\in\R^d:\|\mathbf v\|\leq1\}, & \mathbf x=\mathbf0,\ \alpha=0.
\end{cases}
\]
\subsection{Convex analysis}
\label{subsec:convex}

For a convex function \(U:\mathcal X\to\R\), denote its conjugate by
\begin{equation}
\label{eq:conjugate}
U^*(\mathbf y)
:=
\sup_{\mathbf x\in\mathcal X}
\left\{
\mathbf x^\top\mathbf y-U(\mathbf x)
\right\}.
\end{equation}
The Fenchel--Young gap associated with \(U\) is
\begin{equation}
\label{eq:gap}
F_U(\mathbf{x},\mathbf{y}):=U(\mathbf{x})+U^*(\mathbf{y})-\mathbf{x}^\top\mathbf{y}.
\end{equation}
By the Fenchel--Young inequality \cite[Proposition~13.13]{bauschke2011}, \(F_U\geq0\), and if \(U\) is differentiable, \cite[Theorem~16.23]{bauschke2011} gives \(F_U(\mathbf x,\mathbf y)=0\) if and only if \(\mathbf y=\nabla U(\mathbf x)\). It thus quantifies the mismatch between \(\mathbf x\) and \(\mathbf y\) through \(\nabla U\), without requiring the two to live in comparable coordinates, which makes it a natural building block for the Lyapunov function of Section~\ref{sec:proof:sec}. 

Moreover, if \(U:\mathcal X\to\R\) is differentiable and strictly convex, its gradient is strictly monotone \cite[Chapter~22]{bauschke2011}, i.e.,
\begin{equation}
\label{eq:strict-monotone}
(\mathbf x-\mathbf y)^\top
\left(
\nabla U(\mathbf x)-\nabla U(\mathbf y)
\right)
>0,
\ \
\forall\,\mathbf x,\mathbf y\in\mathcal X,\ \mathbf x\neq\mathbf y.
\end{equation}

\subsection{Homogeneity}
\label{subsec:homogeneity}

\begin{definition}[Homogeneity]
Given $\mathcal{X}\subseteq \mathbb{R}^n$, a function \(U:\mathcal X\to\R\) is \emph{homogeneous of degree \(d\in\mathbb{R}\)} if
\(U(\lambda\mathbf x)=\lambda^dU(\mathbf x)\) for every \(\mathbf x\in\mathcal X\) and \(\lambda>0\).
Likewise, a set-valued map \(\mathcal S:\mathcal X\rightrightarrows\mathcal X\) is homogeneous of degree \(d\in\mathbb{R}\) if
\(\mathcal S(\lambda\mathbf x)=\lambda^d\mathcal S(\mathbf x)\).
\end{definition}

For a fixed integer \( m\geq 1\), define the homogeneity weights
\begin{equation}
\label{eq:hom:weights}
r_\mu:=\frac{m+1-\mu}{m+1},
\qquad
\mu=0,\ldots,m.
\end{equation}
For \(\mathbf x=\operatorname{col}(\mathbf x_0,\ldots,\mathbf x_m)\in\mathcal X^{m+1}\) and \(\lambda>0\), define the weighted dilation
\begin{equation}
\label{eq:weighted-dilation}
\delta(\mathbf x;\lambda)
:=
\operatorname{col}
\left(
\lambda^{r_0}\mathbf x_0,\ldots,
\lambda^{r_m}\mathbf x_m
\right).
\end{equation}
The following is the standard definition of weighted homogeneity adapted to the concrete weights $r_\mu$ used in this work.
\begin{definition}[Weighted homogeneity]
Given $\mathcal{X}\subseteq \mathbb{R}^n$, a function \(V:\mathcal X^{m+1}\to\R\) is \emph{weighted homogeneous of degree \(d\in\mathbb{R}\)} if
\(V(\delta(\mathbf x;\lambda))=\lambda^dV(\mathbf x)\) for every \(\mathbf x\in\mathcal X^{m+1}\) and \(\lambda>0\).
Likewise, a set-valued map
\(\mathcal F:\mathcal X^{m+1}\rightrightarrows\mathcal X^{m+1}\)
is weighted homogeneous of degree \(d\) if
\(
\mathcal F(\delta(\mathbf x;\lambda))
=
\left\{
\lambda^d\delta(\mathbf v;\lambda):
\mathbf v\in\mathcal F(\mathbf x)
\right\}
\)
for every \(\mathbf x\in\mathcal X^{m+1}\) and \(\lambda>0\).
\end{definition}

Finally, define the homogeneous gauge characterized by \(\varrho(\delta(\mathbf x;\lambda))=\lambda\varrho(\mathbf x)\) and expressed as:
\begin{equation}
\label{eq:hom:norm}
\varrho(\mathbf x)
:=
\sum_{\mu=0}^m
\|\mathbf x_\mu\|^{1/r_\mu}.
\end{equation}

\section{Abstract homogeneous chains}

We now introduce the class of dynamical systems that provides the common representation used throughout this work. Consider the state \(\mathbf{q}=\operatorname{col}(\mathbf{q}_0,\cdots,\mathbf{q}_m)\in\mathcal{X}^{m+1}\), \(m\geq1\), where \(\mathcal{X}\subseteq\mathbb{R}^{n}\) is a linear subspace. We consider differential inclusions of the form
\begin{equation}
\begin{aligned}
\dot{\mathbf{q}}_\mu&=-k_\mu L^{1-r_{\mu+1}}\mathbf{Q}\nabla U_\mu(\mathbf{q}_0)+\mathbf{q}_{\mu+1}-\gamma r_\mu\mathbf{q}_\mu,\\ &\mu=0,\ldots,m-1,\\
\dot{\mathbf{q}}_m&\in-L\mathbf{Q}\left(k_m\mathcal{S}(\mathbf{q}_0)-\mathcal{D}\right)-\gamma r_m\mathbf{q}_m,
\end{aligned}
\label{eq:system}
\end{equation}
where \(\mathbf{Q}\in\mathbb{R}^{n\times n}\) satisfies \(\mathbf{Q}\mathcal{X}\subseteq\mathcal{X}\), \(U_0,\dots,U_{m-1}:\mathcal{X}\to\R_{\geq0}\) are differentiable functions, \(\mathcal{S}:\mathcal{X}\rightrightarrows\mathcal{X}\) is a set-valued map, and \(\mathcal{D}\subset\mathcal{X}\) is the set of admissible disturbances. The parameters satisfy \(L>0\), \(\gamma\geq0\), and \(k_0,\dots,k_m>0\), and the homogeneity weights $r_\mu$ are defined in \eqref{eq:hom:weights}.\footnote{A solution of \eqref{eq:system} on an interval \([0,T)\) is an absolutely continuous function \(\mathbf{q}:[0,T)\to\mathcal{X}^{m+1}\) for which \eqref{eq:system} holds for almost every \(t\in[0,T)\). }

The following definition identifies the structural properties of \eqref{eq:system} that characterize the class considered in this work.

\begin{definition}[Abstract Homogeneous Chain]
\label{def:ahc}
We say that \eqref{eq:system} is an AHC if the following conditions hold:
\begin{enumerate}
\item \(\mathbf{Q}+\mathbf{Q}^\top\) is positive definite on \(\mathcal X\).
\item For each \(\mu=0,\ldots,m-1\), the function \(U_\mu:\mathcal{X}\to\R_{\geq0}\) is continuously differentiable, convex, positive definite, and homogeneous of degree \(1+r_{\mu+1}=1+\frac{m-\mu}{m+1}\), i.e.,
\begin{equation}
\label{eq:homogenenity:U}
U_\mu(\lambda\mathbf{q})=\lambda^{1+r_{\mu+1}}U_\mu(\mathbf{q})
\end{equation}
for every \(\mathbf{q}\in\mathcal{X}\) and every \(\lambda\geq0\). Its conjugate \(U_\mu^*\) is twice continuously differentiable.
\item The map \(\mathcal{S}:\mathcal{X}\rightrightarrows\mathcal{X}\) is upper semicontinuous and has nonempty, compact, and convex values. It is homogeneous of degree $0$ on \(\mathcal X\). Moreover, there exists \(c_{\mathcal{S}}>0\) such that for every \(\mathbf{q}\in\mathcal{X}\), \(\mathbf{s}\in\mathcal{S}(\mathbf{q})\) then,
\[
\mathbf{q}^\top\mathbf{s}\geq c_{\mathcal{S}}\|\mathbf{q}\|.
\]
\item The disturbance set \(\mathcal{D}\subset\mathcal{X}\) is nonempty, compact, and convex,\footnote{Following from \cite{filippov1988differential}, together with the corresponding conditions on \(\mathcal S\), these assumptions ensure existence of solutions to \eqref{eq:system}.} and \(\|\mathbf{d}\|\leq1\) for every \(\mathbf{d}\in\mathcal{D}\).
\end{enumerate}
\end{definition}

The conditions in Definition~\ref{def:ahc} balance two requirements: they provide the structure needed for the Lyapunov construction developed in Section~\ref{sec:proof:sec}, while remaining broad enough to encompass both classical and networked HOSM dynamics.

To make this abstraction concrete, consider the standard scalar robust exact differentiator as a first example. Choosing the AHC parameters according to the first row of Table~\ref{tab:applications}, with \(\mathcal D=[-1,1]\) and \(\gamma=0\), yields
\[
\begin{aligned}
\dot q_\mu&=-k_\mu L^{1-r_{\mu+1}}\left\lceil q_0\right\rfloor^{r_{\mu+1}}+q_{\mu+1},\qquad \mu=0,\ldots,m-1,\\
\dot q_m&\in-Lk_m\operatorname{sign}(q_0)+[-L,L],
\end{aligned}
\]
with \(q_0(t),\dots,q_m(t)\in\mathbb{R}\). These are precisely the error dynamics of the arbitrary-order robust exact differentiator in~\cite{levant2003}. Proposition~\ref{prop:power-potential} in Appendix~\ref{sec:auxiliary}, applied with \(\mathcal X=\R\), \(\mathbf M=1\), and \(r=r_{\mu+1}\), verifies condition~2 of Definition~\ref{def:ahc} for this choice of \(U_0,\dots,U_{m-1}\). Thus, the scalar robust exact differentiator is an AHC. Section~\ref{sec:applications} shows that the same structure also arises in substantially different multi-agent systems.

The central stability result for AHCs is the following.
\begin{theorem}
\label{thm:main}
Let \(m\geq1\) be an integer, \(L>0\), and \(\gamma\geq0\). Let \(\mathbf{Q}\), \(U_\mu\), \(\mathcal S\), and \(\mathcal D\) satisfy the conditions in Definition~\ref{def:ahc}. Then, there exist gains \(k_0,\ldots,k_m>0\) such that the origin of \eqref{eq:system} is globally finite-time stable.\footnote{Global finite-time stability means that the origin is Lyapunov stable and every solution reaches the origin in finite time~\cite{Bhat2005}.}
\end{theorem}
The proof of Theorem~\ref{thm:main} is given in Section~\ref{sec:proof:sec}.

\section{Applications}
\label{sec:applications}
We illustrate Theorem~\ref{thm:main} through three distributed observer problems involving a team of \(N\) agents. The team is able to interact through a communication network modeled by a fixed connected undirected graph \(\mathcal G=(\mathcal V,\mathcal E)\), with node set \(\mathcal V=\{1,\ldots,N\}\) and edge set $\mathcal{E}\subseteq \mathcal{V}\times \mathcal{V}$ defining neighbor sets \(\mathcal N_i\subseteq\mathcal{V}\). Assign an arbitrary orientation to the edges and let
\(\mathbf D=[\mathbf d_1,\ldots,\mathbf d_{|\mathcal E|}]\in\R^{N\times|\mathcal E|}\)
be the corresponding incidence matrix. Each column \(\mathbf d_\ell\in\mathbb{R}^N\) corresponds to an edge, with one entry equal to \(1\) at its tail, one equal to \(-1\) at its head, and zeros elsewhere. The graph Laplacian is \(\mathbf L:=\mathbf D\mathbf D^\top\). 

For each problem, the proposed algorithm yields an error system of the form
\eqref{eq:system} for specific choices of \(\mathcal X\), \(\mathbf{Q}\), \(U_\mu\), and \(\mathcal S\) summarized in Table~\ref{tab:applications}. The following subsections derive these error systems and state the corresponding
consequences of Theorem~\ref{thm:main}. The proofs of the resulting corollaries
are collected in Appendix~\ref{sec:proof:corollaries} for readability.
	
\begin{table*}[t]
\centering
\caption{\emph{AHC representation of different applications, including the multi-agent systems of Section~\ref{sec:applications}. Here, \([\mathbf x]_i\) denotes the \(i\)-th \(d\)-dimensional block of \(\mathbf x\), with \(d=1\) for leader--follower distributed differentiation and \(d\) the ambient-space dimension for affine formation tracking. The reported values of \(c_{\mathcal S}\) can be computed directly from the corresponding matrices.}}
\label{tab:applications}
\renewcommand{\arraystretch}{1.4}
\begin{tabular}{c|c|c|c|c|c}
\hline
Problem
& \(\mathcal X\)
& \(\mathbf{Q}\)
& \(U_\mu(\mathbf q)\)
& \(\mathbf{Q}\nabla U_\mu(\mathbf q)\)
& \(\mathcal S(\mathbf q)\)
\\
\hline

\begin{tabular}{c}
Robust exact\\
differentiator\\
\end{tabular}
&
\(\R\)
&
\(1\)
&
\(\displaystyle
\frac{1}{1+r_{\mu+1}}
|q|^{1+r_{\mu+1}}
\)
&
\(\left\lceil q\right\rfloor^{r_{\mu+1}}\)
&
\(\begin{gathered}
\operatorname{sign}(q)\\
c_\mathcal{S}=1
\end{gathered}\)
\\
\hline

\begin{tabular}{c}
Leader--follower\\ distributed\\
differentiation
\end{tabular}
&
\(\R^N\)
&
\((\mathbf L+\mathbf B)^{-1}\)
&
\(\displaystyle
\frac{1}{1+r_{\mu+1}}
\sum_{i=1}^N
|[(\mathbf L+\mathbf B)\mathbf q]_i|^{1+r_{\mu+1}}
\)
&
\(\left\lceil
(\mathbf L+\mathbf B)\mathbf q
\right\rfloor^{r_{\mu+1}}\)
&
\(\begin{gathered}
(\mathbf L+\mathbf B)
\operatorname{sign}((\mathbf L+\mathbf B)\mathbf q)\\
c_{\mathcal S}=\lambda_{\min}(\mathbf L+\mathbf B)
\end{gathered}\)
\\
\hline

\begin{tabular}{c}
Leaderless\\ distributed\\
differentiation
\end{tabular}
&
\(\mathbf1^\perp\)
&
\(\mathbf I_N\)
&
\(\displaystyle
\frac{1}{1+r_{\mu+1}}
\sum_{\ell=1}^{|\mathcal E|}
|\mathbf d_\ell^\top\mathbf q|^{1+r_{\mu+1}}
\)
&
\(\mathbf D\left\lceil\mathbf D^\top\mathbf q\right\rfloor^{r_{\mu+1}}\)
&
\(\begin{gathered}
\mathbf D\operatorname{sign}(\mathbf D^\top\mathbf q)\\
c_{\mathcal S}=\sqrt{\lambda_2(\mathbf L)}
\end{gathered}\)
\\
\hline

\begin{tabular}{c}
Multi-leader affine\\
formation tracking
\end{tabular}
&
\(\R^{dN_F}\)
&
\(\mathbf\Omega_{FF}^{-1}\)
&
\(\displaystyle
\frac{1}{1+r_{\mu+1}}
\sum_{i=1}^{N_F}
\left\|[\mathbf\Omega_{FF}\mathbf q]_i\right\|^{1+r_{\mu+1}}
\)
&
\(\displaystyle
\operatorname{col}_{i=1}^{N_F}
\left(
\vsgn{[\mathbf\Omega_{FF}\mathbf q]_i}{r_{\mu+1}}
\right)
\)
&
\(\begin{gathered}
\displaystyle
\mathbf\Omega_{FF}
\operatorname{col}_{i=1}^{N_F}
\left(
\vsgn{[\mathbf\Omega_{FF}\mathbf q]_i}{0}
\right)\\
c_{\mathcal S}=\lambda_{\min}(\mathbf\Omega_{FF})
\end{gathered}\)
\\
\hline
\end{tabular}
\end{table*}

\subsection{Leader-follower distributed differentiation}
\label{subsec:leader}

In the leader-follower distributed differentiation problem we assume \(z_0(t)\in\R\) is a leader signal available only to a nonempty subset of agents. For each \(i\in\mathcal V\), let \(b_i=1\) if agent \(i\) has access to \(z_0(t)\), and \(b_i=0\) otherwise, and define \(\mathbf B:=\operatorname{diag}(b_1,\ldots,b_N)\). The objective is for every agent to generate derivative estimates \(\{\hat z_{i,\mu}(t)\}_{\mu=0}^m\) such that
\begin{equation}
\label{eq:leader:objective}
\lim_{t\to\infty}\left|\hat z_{i,\mu}(t)-z_0^{(\mu)}(t)\right|=0,\qquad \mu=0,\ldots,m.
\end{equation}
For this purpose, each agent implements
\begin{equation}
\label{eq:leader:protocol}
\begin{aligned}
\dot{\hat z}_{i,\mu}
&=-k_\mu L^{1-r_{\mu+1}}
\left\lceil
\sum_{j\in\mathcal N_i}(\hat z_{i,0}-\hat z_{j,0})
+b_i(\hat z_{i,0}-z_0)
\right\rfloor^{r_{\mu+1}}\\
&+\hat z_{i,\mu+1}, \quad\mu=0,\ldots,m-1,\\
\dot{\hat z}_{i,m}
&\in-k_mL
\operatorname{sign}\left(
\sum_{j\in\mathcal N_i}(\hat z_{i,0}-\hat z_{j,0})
+b_i(\hat z_{i,0}-z_0)
\right),
\end{aligned}
\end{equation}
with parameters $L,k_0,\dots,k_m$. Only \(\hat z_{i,0}(t)\) is exchanged among neighboring agents. Convergence of \eqref{eq:leader:protocol} requires the following assumption.

\begin{assumption}
\label{ass:leader}
There exists \(L>0\) such that for every \(t\geq0\),
\[
|z_0^{(m+1)}(t)|\leq\frac{L}{\sqrt N}.
\]
\end{assumption}

To study convergence, let \(\hat{\mathbf z}_\mu:=\operatorname{col}(\hat z_{1,\mu},\ldots,\hat z_{N,\mu})\) and
\begin{equation}
\label{eq:leader:error:coordinates}
\mathbf{q}_\mu:=\hat{\mathbf z}_\mu-\mathbf1 z_0^{(\mu)},\qquad \mu=0,\ldots,m.
\end{equation}
Moreover, define
\[
e_i:=\sum_{j\in\mathcal N_i}(\hat z_{i,0}-\hat z_{j,0})+b_i(\hat z_{i,0}-z_0),
\]
so that, with \(\mathbf e:=\operatorname{col}(e_1,\ldots,e_N)\), one has \[\mathbf e=(\mathbf L+\mathbf B)\mathbf{q}_0.\] Therefore, the error dynamics are
\begin{equation}
\label{eq:leader:error}
\begin{aligned}
\dot{\mathbf{q}}_\mu
&=-k_\mu L^{1-r_{\mu+1}}
\left\lceil(\mathbf L+\mathbf B)\mathbf{q}_0\right\rfloor^{r_{\mu+1}}
+\mathbf{q}_{\mu+1},\\
&\mu=0,\ldots,m-1,\\
\dot{\mathbf{q}}_m
&\in-k_mL\operatorname{sign}((\mathbf L+\mathbf B)\mathbf{q}_0)
-\mathbf1 z_0^{(m+1)}(t).
\end{aligned}
\end{equation}
With appropriate choices summarized in Table~\ref{tab:applications}, \eqref{eq:leader:error} is of the form of the AHC in \eqref{eq:system}, so Theorem~\ref{thm:main} can be used to conclude the following.

\begin{corollary}
\label{cor:leader}
Suppose that Assumption~\ref{ass:leader} holds and that at least one agent has access to the leader signal. Then, there exist gains \(k_0,\ldots,k_m>0\) such that the origin of \eqref{eq:leader:error} is globally finite-time stable. Consequently, \eqref{eq:leader:protocol} achieves \eqref{eq:leader:objective} for every initial condition.
\end{corollary}

\begin{remark}
The protocol \eqref{eq:leader:protocol} was introduced in \cite{elf}, where stability was established for arbitrary order \(m\). However, a Lyapunov function and numerical gain-proposal rules were provided only for \(m=1\). In contrast, Section~\ref{sec:gain-selection} provides a recursive numerical optimization-based gain-proposal procedure for general AHCs with \eqref{eq:leader:protocol} as a particular case.
\end{remark}

\subsection{Leaderless distributed differentiation}
\label{subsec:redcho}

A leaderless distributed differentiator aims to solve a dynamic average consensus problem where each agent \(i\in\mathcal{V}\) has access to a signal \(z_i(t)\in\mathbb{R}\) which may correspond to the output of a local sensor. The objective is for each agent to generate derivative estimates \(\{\hat z_{i,\mu}(t)\}_{\mu=0}^m\) such that
\begin{equation}
\label{eq:distro:diff}
\lim_{t\to\infty}\left|\hat z_{i,\mu}(t)-\bar z^{(\mu)}(t)\right|=0,\qquad \bar z(t):=\frac{1}{N}\sum_{i=1}^N z_i(t).
\end{equation}
To solve \eqref{eq:distro:diff}, each agent \(i\) introduces the auxiliary variables \(\{\upeta_{i,\mu}(t)\}_{\mu=0}^m\) and implements the REDCHO protocol \cite{redcho}:
\begin{equation}
\label{eq:redcho}
\begin{aligned}
\dot{\upeta}_{i,\mu}&=k_\mu L^{1-r_{\mu+1}}\sum_{j\in\mathcal N_i}\left\lceil\hat z_{i,0}-\hat z_{j,0}\right\rfloor^{r_{\mu+1}}+\upeta_{i,\mu+1}-\gamma r_\mu\upeta_{i,\mu},\\&\mu=0,\ldots,m-1,\\
\dot{\upeta}_{i,m}&\in k_mL\sum_{j\in\mathcal N_i}\operatorname{sign}(\hat z_{i,0}-\hat z_{j,0})-\gamma r_m\upeta_{i,m},\\
\hat z_{i,\mu}&=z_i^{(\mu)}-\sum_{\nu=0}^{m}[\mathbf G]_{\mu+1,\nu+1}\upeta_{i,\nu},\qquad \mu=0,\ldots,m,
\end{aligned}
\end{equation}
with parameters $L,k_0,\dots,k_m$, and fixed \(\gamma\geq 0\),
\[
\mathbf{\Gamma}=\begin{bmatrix}
-\gamma r_0 & 1 & 0 & \cdots & 0\\
0 & -\gamma r_1 & 1 & \ddots & \vdots\\
\vdots & \ddots & \ddots & \ddots & 0\\
0 & \cdots & 0 & -\gamma r_{m-1} & 1\\
0 & \cdots & \cdots & 0 & -\gamma r_m
\end{bmatrix},
\]
and \(\mathbf G=\operatorname{col}(\mathbf C,\mathbf C\mathbf{\Gamma},\ldots,\mathbf C\mathbf{\Gamma}^{m})\), with \(\mathbf C=[1\ 0\ \cdots\ 0]\). The matrix \(\mathbf G\) is nonsingular since it is the observability matrix of the observable pair \((\mathbf\Gamma,\mathbf C)\). Since its first row is \(\mathbf C\), only \(\hat z_{i,0}(t)=z_i(t)-\upeta_{i,0}(t)\) is exchanged. 
Convergence of \eqref{eq:redcho} requires the following assumption.

\begin{assumption}
\label{ass:redcho}
Given $\gamma\geq 0$, let \(\ell_0,\ldots,\ell_m\) be defined by the polynomial \(q(s)=\prod_{\mu=0}^{m}(s+\gamma r_\mu)=s^{m+1}+\sum_{\mu=0}^{m}\ell_\mu s^\mu\). Assume that there exists \(L>0\) such that \[|v_i(t)|\leq \frac{L}{\sqrt{N}}\] for every \(i=1,\ldots,N\) and \(t\geq0\) where
\[
v_i(t)=z_i^{(m+1)}(t)-\bar z^{(m+1)}(t)+\sum_{\mu=0}^{m}\ell_\mu\left(z_i^{(\mu)}(t)-\bar z^{(\mu)}(t)\right).
\]

\end{assumption}

To study convergence, let \(\hat{\mathbf z}_\mu:=\operatorname{col}(\hat z_{1,\mu},\ldots,\hat z_{N,\mu})\), \(\mathbf P:=\mathbf I_N-\mathbf1\mathbf1^\top/N\), and
\begin{equation}
\label{eq:redcho:disagreement}
\operatorname{col}
(\mathbf{q}_0,\ldots,\mathbf{q}_m)
=
(\mathbf G^{-1}\otimes\mathbf P)
\operatorname{col}
(\hat{\mathbf z}_0,\ldots,\hat{\mathbf z}_m).
\end{equation}
In particular, \(\mathbf{q}_0=\mathbf P\hat{\mathbf z}_0\), and for every \(\mu=0,\ldots,m\),
\[\mathbf{q}_\mu\in\mathbf1^\perp:=\{\mathbf x\in\R^N:\mathbf1^\top\mathbf x=0\}.\] It was shown in \cite{redcho} that convergence of the variables in \eqref{eq:redcho:disagreement} to the origin implies \eqref{eq:distro:diff}. Thus, it remains to establish stability of their dynamics, which a direct computation given in detail in \cite{redcho} leads to

\begin{equation}
\label{eq:redcho:error}
\begin{aligned}
\dot{\mathbf{q}}_\mu&=-k_\mu L^{1-r_{\mu+1}}\mathbf D\left\lceil\mathbf D^\top\mathbf{q}_0\right\rfloor^{r_{\mu+1}}+\mathbf{q}_{\mu+1}-\gamma r_\mu\mathbf{q}_\mu,\\ &\mu=0,\ldots,m-1,\\
\dot{\mathbf{q}}_m&\in-L\left(k_m\mathbf D\operatorname{sign}(\mathbf D^\top\mathbf{q}_0)-\mathbf d(t)\right)-\gamma r_m\mathbf{q}_m,
\end{aligned}
\end{equation}
where \(\mathbf d(t):=\mathbf v(t)/L\), with \(\mathbf v(t):=\operatorname{col}(v_1(t),\ldots,v_N(t))\), and the set-valued sign in \(\mathbf D\operatorname{sign}(\mathbf D^\top\mathbf q_0)\) is understood with one selection per edge. By Assumption~\ref{ass:redcho}, \(\|\mathbf d(t)\|\leq1\).

With appropriate choices summarized in Table~\ref{tab:applications}, the error dynamics \eqref{eq:redcho:error} are of the form of the AHC \eqref{eq:system}, so Theorem \ref{thm:main} can be used to conclude the following.

\begin{corollary}
\label{cor:redcho}
Suppose that Assumption~\ref{ass:redcho} holds and \(\gamma\geq0\). Then, there exist gains \(k_0,\ldots,k_m>0\) such that the origin of \eqref{eq:redcho:error} is globally finite-time stable. Moreover, if \(\gamma>0\), \eqref{eq:redcho} achieves \eqref{eq:distro:diff} for every initial condition, while if \(\gamma=0\), the same conclusion holds provided that
\begin{equation}
    \label{eq:initial}
    \sum_{i=1}^N \upeta_{i,\mu}(0) = 0, \mu=0,\dots,m.
\end{equation}
\end{corollary}

\begin{remark}
For \(\gamma=0\), \eqref{eq:redcho} reduces to the EDCHO protocol in \cite{edcho}, whose gain-design procedure leaves one gain to be selected sufficiently large without an arbitrary-order expression for the threshold. In contrast, Section~\ref{sec:gain-selection} provides a numerical optimization-based gain-proposal procedure for the general class of AHCs with EDCHO as a particular case. For \(\gamma>0\), Corollary~\ref{cor:redcho} also extends the REDCHO result in \cite{redcho} from local to global convergence.
\end{remark}

\subsection{Multi-leader affine formation tracking}
\label{subsec:affine}

We consider the affine formation tracking problem in \cite{zhao2018}. The \(N>N_L\) agents are partitioned into \(N_L=d+1\) leaders \(\mathcal I_L=\{1,\dots,N_L\}\) and \(N_F=N-N_L\) followers \(\mathcal I_F=\{N_L+1,\dots,N\}\), and \(\mathbf p_i(t)\in\R^d\) denotes the position of agent \(i\). Given a nominal configuration \(\{\mathbf p_i^\star\}_{i=1}^N\), the team must realize a time-varying affine transformation of this configuration. The leaders select which one. They move freely, and each follower must be in the position consistent with the current leader positions.

Two properties of the nominal configuration turn this requirement into an algebraic one. The first is that the admissible formations form a linear space. Let \(\{\omega_{ij}\}_{(i,j)\in\mathcal E}\) be symmetric stress parameters, i.e., \(\omega_{ij}=\omega_{ji}\), and define the corresponding stress matrix \(\mathbf\Omega\in\R^{N\times N}\) by
\[
[\mathbf\Omega]_{ij}
=
\begin{cases}
- \omega_{ij}, & i\neq j,\\[1mm]
\displaystyle\sum_{k\in\mathcal N_i} \omega_{ik}, & i=j.
\end{cases}
\]
Assume that the stress parameters form an equilibrium stress for the nominal configuration, i.e., \(\sum_{j\in\mathcal N_i}\omega_{ij}(\mathbf p_i^\star-\mathbf p_j^\star)=\mathbf0\) for every \(i\), and are selected according to \cite{zhao2018} so that \(\mathbf\Omega\) is positive semidefinite, the nominal configuration affinely spans \(\R^d\), and \(\operatorname{rank}(\mathbf\Omega)=N-d-1\). Then,
\[
\begin{aligned}
&\operatorname{Null}(\mathbf\Omega\otimes\mathbf I_d)
\equiv
\big\{
\operatorname{col}(\mathbf p_1,\ldots,\mathbf p_N)\in\R^{dN}\\
&\qquad\quad:
\mathbf p_i=\mathbf A\mathbf p_i^\star+\mathbf b,\ 
\mathbf A\in\R^{d\times d},\ 
\mathbf b\in\R^d,\ 
i=1,\ldots,N
\big\},
\end{aligned}
\]
where \(\operatorname{Null}(\bullet)\) denotes the nullspace. That is, the admissible formations are the configurations annihilated by the $\mathbf{\Omega}\otimes \mathbf{I}_d$.

The second property is that the leaders determine the followers uniquely. Let \(\mathbf p_L:=\operatorname{col}(\mathbf p_1,\ldots,\mathbf p_{N_L})\), \(\mathbf p_F:=\operatorname{col}(\mathbf p_{N_L+1},\ldots,\mathbf p_N)\), and partition the lifted stress matrix accordingly as
\[
\mathbf\Omega\otimes\mathbf I_d
=
\begin{bmatrix}
\mathbf\Omega_{LL} & \mathbf\Omega_{LF}\\
\mathbf\Omega_{LF}^{\top} & \mathbf\Omega_{FF}
\end{bmatrix}.
\]
If the nominal leader positions are affinely independent, then \(\mathbf\Omega_{FF}\) is symmetric positive definite \cite{zhao2018}. Moreover, consistent configurations belong to \(\operatorname{Null}(\mathbf\Omega\otimes\mathbf I_d)\), so its leader and follower components satisfy
\[
\mathbf\Omega_{LF}^{\top}\mathbf p_L+\mathbf\Omega_{FF}\mathbf p_F=\mathbf0.
\]
Since \(\mathbf\Omega_{FF}\) is invertible, any leader configuration \(\mathbf p_L(t)\) fixes the desired follower configuration through
\[
\mathbf p_{F,\mathrm{ref}}(t):=-\mathbf\Omega_{FF}^{-1}\mathbf\Omega_{LF}^{\top}\mathbf p_L(t).
\]

The reference \(\mathbf p_{F,\mathrm{ref}}(t)\) is therefore known in closed form, but not to the followers. Evaluating it requires the leader positions and the global stress matrix, whereas each follower only interacts with its neighbors.

The problem is thus one of distributed estimation. We construct, for every follower \(i\in\mathcal I_F\), an observer with states \(\{\hat{\mathbf p}_{i,\mu}\}_{\mu=0}^m\) such that
\begin{equation}
\label{eq:affine:objective}
\lim_{t\to\infty}
\left\|
\hat{\mathbf p}_{F,\mu}(t)
-
\mathbf p_{F,\mathrm{ref}}^{(\mu)}(t)
\right\|=0,
\qquad
\mu=0,\ldots,m,
\end{equation}
where \(\hat{\mathbf p}_{F,\mu}:=\operatorname{col}(\hat{\mathbf p}_{i,\mu})_{i\in\mathcal I_F}\). Each follower thus reconstructs not only the reference itself, but also its first \(m\) derivatives for its potential use in a tracking controller. For \(j\in\mathcal I_L\), set \(\hat{\mathbf p}_{j,0}:=\mathbf p_j\). Then, each follower \(i\in\mathcal I_F\) implements
\begin{equation}
\label{eq:affine:observer}
\begin{aligned}
\dot{\hat{\mathbf p}}_{i,\mu}
&=
-k_\mu L^{1-r_{\mu+1}}
\vsgn{\sum_{j\in\mathcal N_i}
\omega_{ij}
(\hat{\mathbf p}_{i,0}-\hat{\mathbf p}_{j,0})
}{r_{\mu+1}}
\\
&+\hat{\mathbf p}_{i,\mu+1},\qquad\mu=0,\ldots,m-1,\\
\dot{\hat{\mathbf p}}_{i,m}
&\in
-k_mL
\vsgn{
\sum_{j\in\mathcal N_i}
\omega_{ij}
(\hat{\mathbf p}_{i,0}-\hat{\mathbf p}_{j,0})
}{0},
\end{aligned}
\end{equation}
with parameters $L,k_0,\dots,k_m$. The only quantity a follower transmits is \(\hat{\mathbf p}_{i,0}(t)\), and the leaders broadcast their positions \(\mathbf p_i(t)\). Convergence of \eqref{eq:affine:observer} requires the following assumption.

\begin{assumption}
\label{ass:affine}
There exists \(L>0\) such that
\[
\left\|
\mathbf\Omega_{LF}^{\top}
\mathbf p_L^{(m+1)}(t)
\right\|
\leq L
\]
for every \(t\geq0\).
\end{assumption}

The assumption bounds how fast the leaders may maneuver. It limits the \((m+1)\)-th derivative of the reference the followers are asked to track. To study convergence, define the error
\begin{equation}
\label{eq:affine:error:coordinates}
\mathbf{q}_\mu
:=
\hat{\mathbf p}_{F,\mu}
-
\mathbf p_{F,\mathrm{ref}}^{(\mu)}
=
\hat{\mathbf p}_{F,\mu}
+
\mathbf\Omega_{FF}^{-1}
\mathbf\Omega_{LF}^{\top}
\mathbf p_L^{(\mu)},
\quad
\mu=0,\ldots,m.
\end{equation}
The disagreement signals driving \eqref{eq:affine:observer} can be stacked to obtain
\[
\begin{aligned}
\operatorname{col}_{i\in\mathcal I_F}
\left(
\sum_{j\in\mathcal N_i}
\omega_{ij}
(\hat{\mathbf p}_{i,0}-\hat{\mathbf p}_{j,0})
\right)
&=
\mathbf\Omega_{LF}^{\top}\mathbf p_L
+
\mathbf\Omega_{FF}\hat{\mathbf p}_{F,0}
\\&=
\mathbf\Omega_{FF}\mathbf q_0.
\end{aligned}
\]
Therefore, the error dynamics are
\begin{equation}
\label{eq:affine:error}
\begin{aligned}
\dot{\mathbf{q}}_\mu
&=
-k_\mu L^{1-r_{\mu+1}}
\operatorname{col}_{i=1}^{N_F}
\left(
\vsgn{
[\mathbf\Omega_{FF}\mathbf q_0]_i
}{r_{\mu+1}}
\right)
+\mathbf{q}_{\mu+1},\\
&\mu=0,\ldots,m-1,\\
\dot{\mathbf{q}}_m
&\in
-k_mL
\operatorname{col}_{i=1}^{N_F}
\left(
\vsgn{
[\mathbf\Omega_{FF}\mathbf q_0]_i
}{0}
\right)
+
\mathbf\Omega_{FF}^{-1}
\mathbf\Omega_{LF}^{\top}
\mathbf p_L^{(m+1)},
\end{aligned}
\end{equation}
where \([\mathbf\Omega_{FF}\mathbf q_0]_i\in\R^d\) denotes the \(i\)-th follower block. With the choices summarized in Table~\ref{tab:applications}, \eqref{eq:affine:error} is of the form of the AHC \eqref{eq:system}, so Theorem~\ref{thm:main} applies.

\begin{corollary}
\label{cor:affine}
Suppose that Assumption~\ref{ass:affine} holds. Then, there exist gains \(k_0,\ldots,k_m>0\) such that the origin of \eqref{eq:affine:error} is globally finite-time stable. Consequently, the observer \eqref{eq:affine:observer} achieves \eqref{eq:affine:objective} for every initial condition.
\end{corollary}

\begin{remark}
The observer \eqref{eq:affine:observer} has not appeared before in the literature. Since each \(\hat{\mathbf p}_{i,\mu}(t)\) provides an estimate of the corresponding derivative of the desired follower trajectory, and these estimates converge in finite time, they can be combined directly with a local trajectory-tracking controller, as in \cite{affine}. In contrast to \cite{affine}, however, \eqref{eq:affine:observer} can also be adapted directly as a controller for first-order agents \(\dot{\mathbf p}_i=\mathbf u_i\). In this case, setting \(\hat{\mathbf p}_{i,0}=\mathbf p_i\), the control input can be chosen as
\[
\mathbf u_i
=
-k_0L^{1-r_1}
\vsgn{
\sum_{j\in\mathcal N_i}
\omega_{ij}
(\mathbf p_i-\mathbf p_j)
}{r_1}
+\hat{\mathbf p}_{i,1},
\]
while the states \(\hat{\mathbf p}_{i,1},\ldots,\hat{\mathbf p}_{i,m}\) evolve according to \eqref{eq:affine:observer}. This implementation requires only relative position measurements $\mathbf{p}_i(t)-\mathbf{p}_j(t)$. Moreover, the operator \(\vsgn{\bullet}{\alpha}\) is rotation equivariant, since \(\vsgn{\mathbf R\mathbf x}{\alpha}=\mathbf R\vsgn{\mathbf x}{\alpha}\) for every \(\mathbf R\in\mathrm{SO}(d)\). Consequently, as opposed to \cite{affine}, the relative measurements, observer states, and control input may all be represented in each agent's fixed local coordinate frame, without requiring a common global frame of reference.
\end{remark}

\section{Proof of the main result}
\label{sec:proof:sec}

The proof is organized in four steps. First, Lemma~\ref{lem:normalization} removes the scale \(L\) and rewrites the dynamics in terms of normalized gains. Second, we construct a homogeneous Lyapunov function from Fenchel--Young gaps associated with the potentials \(U_\mu\) and derive an upper bound for its derivative along solutions. Third, this estimate is used to select the normalized gains recursively and obtain a finite-time Lyapunov inequality for \(\gamma=0\). Finally, the case \(\gamma\geq0\) is obtained from the case \(\gamma=0\) through a time reparameterization.

\subsection{Normalization}

\begin{lemma}
\label{lem:normalization}
Fix any \(\kappa>0\) and consider the change of coordinates
\[
\mathbf x_0:=\frac{\mathbf q_0}{\kappa L},
\qquad
\mathbf x_\mu:=\frac{\mathbf Q^{-1}\mathbf q_\mu}{k_{\mu-1}\kappa^{r_\mu}L},
\qquad
\mu=1,\ldots,m.
\]
Then, \eqref{eq:system} is transformed into
\begin{equation}
\begin{aligned}
\dot{\mathbf x}_0
&=
-\widetilde{k}_0\mathbf Q
\left(
\nabla U_0(\mathbf x_0)-\mathbf x_1
\right)
-\gamma r_0\mathbf x_0,\\
\dot{\mathbf x}_\mu
&=
-\widetilde{k}_\mu
\left(
\nabla U_\mu(\mathbf x_0)-\mathbf x_{\mu+1}
\right)
-\gamma r_\mu\mathbf x_\mu,
\ \ 
\mu=1,\ldots,m-1,\\
\dot{\mathbf x}_m
&\in
-\widetilde{k}_m
\left(
\mathcal S(\mathbf x_0)-\frac{1}{k_m}\mathcal D
\right)
-\gamma r_m\mathbf x_m,
\end{aligned}
\label{eq:normalized-system}
\end{equation}
where
\[
\widetilde{k}_\mu
:=
\begin{cases}
k_0\kappa^{-1/(m+1)}, & \mu=0,\\
\dfrac{k_\mu}{k_{\mu-1}}\kappa^{-1/(m+1)}, & \mu=1,\ldots,m.
\end{cases}
\]
In particular,
\begin{equation}
\label{eq:normalized-gain-product}
k_m
=
\kappa
\prod_{\mu=0}^m\widetilde{k}_\mu.
\end{equation}
Moreover, the origin of \eqref{eq:system} is globally finite-time stable if and only if the origin of \eqref{eq:normalized-system} is globally finite-time stable.
\end{lemma}

\begin{proof}
Since \(U_\mu\) is homogeneous of degree \(1+r_{\mu+1}\), its gradient is homogeneous of degree \(r_{\mu+1}\), so \(\nabla U_\mu(\kappa L\mathbf x_0)=(\kappa L)^{r_{\mu+1}}\nabla U_\mu(\mathbf x_0)\). Substituting \(\mathbf q_0=\kappa L\mathbf x_0\) and \(\mathbf q_\mu=k_{\mu-1}\kappa^{r_\mu}L\mathbf Q\mathbf x_\mu\), \(\mu=1,\ldots,m\), into \eqref{eq:system}, and using \(r_{\mu+1}-r_\mu=-1/(m+1)\), gives \eqref{eq:normalized-system} with the stated normalized gains. The degree-zero homogeneity of \(\mathcal S\) is used in the last inclusion. Multiplying the expressions for \(\widetilde{k}_0,\ldots,\widetilde{k}_m\) gives \eqref{eq:normalized-gain-product}. Since \(\mathbf Q+\mathbf Q^\top\) is positive definite on \(\mathcal X\), the restriction of \(\mathbf Q\) to \(\mathcal X\) is invertible. Hence, the change of coordinates is linear and invertible for every \(\kappa>0\), so it preserves global finite-time stability.
\end{proof}

\subsection{Lyapunov construction}

The Lyapunov function is assembled from Fenchel--Young gaps \eqref{eq:gap} of the potentials. Define the functions \(V_0,\ldots,V_m\) by
\begin{equation}
\label{eq:Vmu}
\begin{aligned}
V_0(\mathbf{x})
&:=
U_0(\mathbf{x}_0)
+
U_0^*(\mathbf{x}_1)
-
\mathbf{x}_0^\top\mathbf{x}_1,\\
V_\mu(\mathbf{x})
&:=
U_{\mu-1}\left(\nabla U_\mu^*(\mathbf{x}_{\mu+1})\right)
+
U_{\mu-1}^*(\mathbf{x}_\mu)-
\mathbf{x}_\mu^\top\nabla U_\mu^*(\mathbf{x}_{\mu+1}),\\&
\mu=1,\ldots,m-1,\\
V_m(\mathbf{x})
&:=
U_{m-1}^*(\mathbf{x}_m).
\end{aligned}
\end{equation}
Note that by \eqref{eq:gap}, \(V_0,\ldots,V_{m-1}\) are Fenchel--Young gaps:
\[
\begin{aligned}
V_0(\mathbf{x})
&=
F_{U_0}(\mathbf x_0,\mathbf x_1),\\
V_\mu(\mathbf{x})
&=
F_{U_{\mu-1}}
\left(
\nabla U_\mu^*(\mathbf x_{\mu+1}),
\mathbf x_\mu
\right),
\qquad
\mu=1,\ldots,m-1.
\end{aligned}
\]
Under the dilation \eqref{eq:weighted-dilation}, \(\mathbf x_{\mu+1}\) has degree \(r_{\mu+1}\). Proposition~\ref{prop:potential-properties}, item~\ref{item:potential-conjugate} in Appendix \ref{sec:auxiliary}, gives degree \(1/r_{\mu+1}\) for \(\nabla U_\mu^*\), so \(\nabla U_\mu^*(\mathbf x_{\mu+1})\) has degree one. Proposition~\ref{prop:potential-properties}, item~\ref{item:potential-gap}, then gives
\begin{align}
V_0(\delta(\mathbf{x};\lambda))
&=
\lambda^{1+r_1}V_0(\mathbf{x}),
\label{eq:V0-homogeneity}\\
V_\mu(\delta(\mathbf{x};\lambda))
&=
\lambda^{1+r_\mu}V_\mu(\mathbf{x}),
\qquad
\mu=1,\ldots,m,
\label{eq:Vmu-homogeneity}
\end{align}
where the case \(\mu=m\) follows from the homogeneity of \(U_{m-1}^*\). These degrees are equalized by
\[
\rho_0:=\frac{2}{1+r_1},
\qquad
\rho_\mu:=\frac{2}{1+r_\mu},
\qquad
\mu=1,\ldots,m,
\]
which satisfy \(\rho_\mu>1\). Fix arbitrary \(\beta_0,\ldots,\beta_m>0\) and consider
\begin{equation}
\label{eq:lyap}
V(\mathbf{x})
:=
\sum_{\mu=0}^m
\frac{\beta_\mu}{\rho_\mu}
V_\mu(\mathbf{x})^{\rho_\mu}.
\end{equation}
Every term \(V_\mu^{\rho_\mu}\) is homogeneous of degree
\begin{equation}
\label{eq:V-degree}
d_V
:=
\rho_0(1+r_1)
=
\rho_\mu(1+r_\mu)
=
2,
\qquad
\mu=1,\ldots,m,
\end{equation}
and therefore
\begin{equation}
\label{eq:V-homogeneity}
V(\delta(\mathbf{x};\lambda))
=
\lambda^{d_V}V(\mathbf{x}).
\end{equation}

\begin{lemma}
\label{lem:lyapunov-properties}
Let \(V_0,\ldots,V_m\) and \(V\) be defined as in \eqref{eq:Vmu} and \eqref{eq:lyap}. Then, the following properties hold.
\begin{enumerate}

\item\label{item:V-zero-sets} Each \(V_\mu\) is nonnegative. Define \[\mathcal C_\mu:=\{\mathbf{x}\in\mathcal{X}^{m+1}:V_0(\mathbf{x})=\cdots=V_\mu(\mathbf{x})=0\}.\] Then, for \(\mu=0,\ldots,m-1\),
\begin{equation}
\label{eq:Cmu-characterization}
\begin{aligned}
\mathcal{C}_\mu
=
\big\{&
\mathbf{x}\in\mathcal{X}^{m+1}: \\&
\mathbf{x}_0
=
\nabla U_0^*(\mathbf{x}_1)
=
\cdots
=
\nabla U_\mu^*(\mathbf{x}_{\mu+1})
\big\},
\end{aligned}
\end{equation}
and
\begin{equation}
\label{eq:Cm-origin}
\mathcal C_m=\{\mathbf0\}.
\end{equation}

\item\label{item:V-positive} The functions \(V_0,\ldots,V_m\) and \(V\) are continuously differentiable, and \(V\) is positive definite and radially unbounded.
\end{enumerate}
\end{lemma}

\begin{proof}
For item~\ref{item:V-zero-sets}, \(V_0,\ldots,V_{m-1}\) are the Fenchel--Young gaps described above. Proposition~\ref{prop:potential-properties}, item~\ref{item:potential-gap}, in Appendix \ref{sec:auxiliary} gives \(V_0\geq0\), with equality if and only if \(\mathbf{x}_0=\nabla U_0^*(\mathbf{x}_1)\), and \(V_\mu\geq0\), with equality if and only if \(\nabla U_{\mu-1}^*(\mathbf{x}_\mu)=\nabla U_\mu^*(\mathbf{x}_{\mu+1})\). Positive definiteness of \(U_{m-1}^*\) gives \(V_m\geq0\), with equality if and only if \(\mathbf{x}_m=\mathbf0\). Combining these equality conditions proves \eqref{eq:Cmu-characterization}. If \(\mathbf{x}\in\mathcal C_m\), then \(\mathbf{x}\in\mathcal C_{m-1}\) and \(\mathbf{x}_m=\mathbf0\), so \eqref{eq:Cmu-characterization} and \(\nabla U_{m-1}^*(\mathbf0)=\mathbf0\) give
\[
\mathbf{x}_0
=
\nabla U_0^*(\mathbf{x}_1)
=
\cdots
=
\nabla U_{m-1}^*(\mathbf{x}_m)
=
\mathbf0.
\]
Applying \(\nabla U_{\mu-1}\) to \(\nabla U_{\mu-1}^*(\mathbf{x}_\mu)=\mathbf0\) gives \(\mathbf{x}_\mu=\mathbf0\) for every \(\mu=1,\ldots,m\). Hence \(\mathbf{x}=\mathbf0\), proving \eqref{eq:Cm-origin}.

For item~\ref{item:V-positive}, condition~2 of Definition~\ref{def:ahc} implies that \(\nabla U_\mu^*\) is continuously differentiable. Hence, \(V_0,\ldots,V_{m-1}\) are continuously differentiable, while \(V_m\) is continuously differentiable by definition. Since \(V_\mu\geq0\) and \(\rho_\mu>1\), the powers in \eqref{eq:lyap} are continuously differentiable. Thus, \(V\) is continuously differentiable. Since \(\beta_0,\ldots,\beta_m>0\), nonnegativity of the terms and \eqref{eq:Cm-origin} show that \(V\) is positive definite. For radial unboundedness, since \(V\) is continuous, positive definite, and homogeneous of degree \(d_V>0\), there exists \(c_V>0\) such that \(V(\mathbf x)\geq c_V\varrho(\mathbf x)^{d_V}\) for every \(\mathbf x\in\mathcal X^{m+1}\), with \(\varrho\) defined in \eqref{eq:hom:norm}. Therefore, \(V\) is radially unbounded.
\end{proof}

\subsection{Derivative along solutions}

We compute the derivative of \(V\) in coordinates adapted to the potentials. Define
\begin{equation}
\label{eq:proof-xi}
\boldsymbol{\xi}_0:=\mathbf x_0,
\qquad
\boldsymbol{\xi}_\mu:=\nabla U_{\mu-1}^*(\mathbf x_\mu),
\quad
\mu=1,\ldots,m,
\end{equation}
and \(\boldsymbol{\xi}
:=
\operatorname{col}(\boldsymbol{\xi}_0,\ldots,\boldsymbol{\xi}_m)\). By Proposition~\ref{prop:potential-properties}, item~\ref{item:potential-conjugate}, in Appendix \ref{sec:auxiliary},
\begin{equation}
\label{eq:proof-x-xi}
\mathbf x_\mu=\nabla U_{\mu-1}(\boldsymbol{\xi}_\mu),
\qquad
\mu=1,\ldots,m.
\end{equation}
Moreover, \eqref{eq:Cmu-characterization} and \eqref{eq:proof-xi} give
\begin{equation}
\label{eq:cmu:xi}
\mathcal C_\mu
=
\left\{
\mathbf x\in\mathcal X^{m+1}:
\boldsymbol{\xi}_0
=
\cdots
=
\boldsymbol{\xi}_{\mu+1}
\right\},
\quad
\mu=0,\ldots,m-1.
\end{equation}
By \eqref{eq:proof-x-xi}, the terms driving \eqref{eq:normalized-system} satisfy \[\nabla U_\mu(\mathbf x_0)-\mathbf x_{\mu+1}=\nabla U_\mu(\boldsymbol{\xi}_0)-\nabla U_\mu(\boldsymbol{\xi}_{\mu+1}).\] Thus, every level compares \(\boldsymbol{\xi}_{\mu+1}\) with \(\boldsymbol{\xi}_0\) through \(\nabla U_\mu\), while the sets \(\mathcal C_\mu\) in \eqref{eq:cmu:xi} are determined by equality of consecutive levels. Throughout, fix
\begin{equation}
\label{eq:km}
k_m>\frac{1}{c_{\mathcal S}}.
\end{equation}

To express \(\dot V\) in a form suitable for the recursive gain argument, we collect the terms arising in its derivative through the following quantities. For \(\mathbf s\in\mathcal S(\boldsymbol{\xi}_0)\) and \(\mathbf d\in\mathcal D\), define

\begin{equation}
\label{eq:proof-auxiliary-functions}
\begin{aligned}
&\mathbf h_\mu(\boldsymbol{\xi})
:=
\nabla U_\mu(\boldsymbol{\xi}_0)
-
\nabla U_\mu(\boldsymbol{\xi}_{\mu+1}),
\qquad
\mu=0,\ldots,m-1,\\
&\boldsymbol{\chi}_\mu(\boldsymbol{\xi};\mathbf s,\mathbf d)
:=
\begin{cases}
\mathbf h_\mu(\boldsymbol{\xi}),
&\mu=1,\ldots,m-1,\\
\mathbf s-\dfrac{1}{k_m}\mathbf d,
&\mu=m,
\end{cases}\\
&\psi_\mu(\boldsymbol{\xi})
:=
-
(\boldsymbol{\xi}_\mu-\boldsymbol{\xi}_{\mu+1})^\top
\left(
\nabla U_\mu(\boldsymbol{\xi}_0)
-
\nabla U_\mu(\boldsymbol{\xi}_\mu)
\right),\\&
\qquad
\mu=1,\ldots,m-1,\\
&\psi_m(\boldsymbol{\xi};\mathbf s,\mathbf d)
:=
-
\boldsymbol{\xi}_m^\top
\boldsymbol{\chi}_m(\boldsymbol{\xi};\mathbf s,\mathbf d),\\
&H_\mu(\boldsymbol{\xi})
:=\\&
\begin{cases}
U_0(\boldsymbol{\xi}_0)
-
U_0(\boldsymbol{\xi}_1)
-
(\boldsymbol{\xi}_0-\boldsymbol{\xi}_1)^\top
\nabla U_0(\boldsymbol{\xi}_1),
\qquad\mu=0,\\
U_{\mu-1}(\boldsymbol{\xi}_{\mu+1})
-
U_{\mu-1}(\boldsymbol{\xi}_\mu)
-
(\boldsymbol{\xi}_{\mu+1}-\boldsymbol{\xi}_\mu)^\top
\nabla U_{\mu-1}(\boldsymbol{\xi}_\mu),\\\quad
\mu=1,\ldots,m-1,\\
r_mU_{m-1}(\boldsymbol{\xi}_m),
\qquad\mu=m,
\end{cases}\\
&\theta_\mu(\boldsymbol{\xi};\mathbf s,\mathbf d;\widetilde{k}_{\mu+1})
:=\\&
\begin{cases}
\widetilde{k}_1
(\boldsymbol{\xi}_0-\boldsymbol{\xi}_1)^\top
\boldsymbol{\chi}_1(\boldsymbol{\xi};\mathbf s,\mathbf d),
\qquad \mu=0,\\
-\widetilde{k}_{\mu+1}
\left(
\nabla U_{\mu-1}(\boldsymbol{\xi}_{\mu+1})
-
\nabla U_{\mu-1}(\boldsymbol{\xi}_\mu)
\right)^\top\\
\times
\nabla^2U_\mu^*
\left(
\nabla U_\mu(\boldsymbol{\xi}_{\mu+1})
\right)
\boldsymbol{\chi}_{\mu+1}(\boldsymbol{\xi};\mathbf s,\mathbf d),
\quad\mu=1,\ldots,m-1.
\end{cases}
\end{aligned}
\end{equation}
and
\begin{equation}
\label{eq:proof-Bnu}
\begin{aligned}
&B_\mu(\boldsymbol{\xi})
:=\\&
\begin{cases}
H_0(\boldsymbol{\xi})^{\rho_0-1}
\mathbf h_0(\boldsymbol{\xi})^\top
\mathbf Q
\mathbf h_0(\boldsymbol{\xi}),
\qquad \mu=0,\\
H_\mu(\boldsymbol{\xi})^{\rho_\mu-1}
(\boldsymbol{\xi}_\mu-\boldsymbol{\xi}_{\mu+1})^\top
\left(
\nabla U_\mu(\boldsymbol{\xi}_\mu)
-
\nabla U_\mu(\boldsymbol{\xi}_{\mu+1})
\right),\\
\quad \mu=1,\ldots,m-1.
\end{cases}
\end{aligned}
\end{equation}

By convexity and positive definiteness of the potentials, the functions \(H_\mu\) are nonnegative. Consequently, \(B_\mu\geq0\) for \(\mu=0,\ldots,m-1\), since \(\mathbf h_0^\top\mathbf Q\mathbf h_0=\frac{1}{2}\mathbf h_0^\top(\mathbf Q+\mathbf Q^\top)\mathbf h_0\geq0\), while Proposition~\ref{prop:potential-properties}, item~\ref{item:potential-primal}, in Appendix \ref{sec:auxiliary} and \eqref{eq:strict-monotone} give the corresponding result for \(\mu=1,\ldots,m-1\). Moreover, \(B_\mu\) and \(\theta_\mu\) vanish when \(\boldsymbol{\xi}_\mu=\boldsymbol{\xi}_{\mu+1}\), while \(\psi_\mu\) vanishes when \(\boldsymbol{\xi}_0=\boldsymbol{\xi}_\mu\) for \(\mu=1,\ldots,m-1\).

\begin{lemma}
\label{lem:lyapunov-upper-bound}
Let \eqref{eq:system} be an AHC with \(\gamma=0\), let \(\widetilde{k}_0,\ldots,\widetilde{k}_m>0\), and let \(V\) be defined by \eqref{eq:lyap}. For \(\mathbf s\in\mathcal S(\boldsymbol{\xi}_0)\) and \(\mathbf d\in\mathcal D\), define
\begin{equation}
\label{eq:proof-Phi-selection}
\begin{aligned}
&\Phi(\boldsymbol{\xi};\mathbf s,\mathbf d)
:=
\beta_0\widetilde{k}_0B_0(\boldsymbol{\xi})
-
\beta_0H_0(\boldsymbol{\xi})^{\rho_0-1}
\theta_0(\boldsymbol{\xi};\mathbf s,\mathbf d;\widetilde{k}_1)+\\
&
\sum_{\mu=1}^{m-1}
\beta_\mu\left[
\widetilde{k}_\mu B_\mu(\boldsymbol{\xi})
-
H_\mu(\boldsymbol{\xi})^{\rho_\mu-1}
\left(
\widetilde{k}_\mu\psi_\mu(\boldsymbol{\xi})
+
\theta_\mu(\boldsymbol{\xi};\mathbf s,\mathbf d;\widetilde{k}_{\mu+1})
\right)
\right]\\
&-
\beta_mH_m(\boldsymbol{\xi})^{\rho_m-1}
\widetilde{k}_m
\psi_m(\boldsymbol{\xi};\mathbf s,\mathbf d).
\end{aligned}
\end{equation}
Set
\begin{equation}
\label{eq:proof-Phi}
\Phi(\boldsymbol{\xi})
:=
\min_{\substack{
\mathbf s\in\mathcal S(\boldsymbol{\xi}_0)\\
\mathbf d\in\mathcal D
}}
\Phi(\boldsymbol{\xi};\mathbf s,\mathbf d).
\end{equation}
Then, every solution \(\mathbf x(t)\) of \eqref{eq:normalized-system} satisfies
\begin{equation}
\label{eq:proof-Phi-upper}
\dot V(\mathbf x(t))
\leq
-\Phi(\boldsymbol{\xi}(\mathbf x(t)))
\end{equation}
for almost every \(t\) in its interval of definition.
\end{lemma}

\begin{proof}
Let \(\mathbf x(t)\) be a solution of \eqref{eq:normalized-system}. Since \(\gamma=0\), for almost every \(t\) there exist measurable selections \(\mathbf s(t)\in\mathcal S(\mathbf x_0(t))\) and \(\mathbf d(t)\in\mathcal D\) such that, using \eqref{eq:proof-x-xi},
\begin{equation}
\label{eq:proof-normalized-dynamics}
\begin{aligned}
\dot{\mathbf x}_0
&=
-\widetilde{k}_0\mathbf Q\mathbf h_0,\\
\dot{\mathbf x}_\mu
&=
-\widetilde{k}_\mu
\boldsymbol{\chi}_\mu(\boldsymbol{\xi};\mathbf s,\mathbf d),
\qquad
\mu=1,\ldots,m.
\end{aligned}
\end{equation}
In terms of \(\boldsymbol{\xi}_\mu\), \eqref{eq:Vmu} becomes
\begin{equation}
\label{eq:proof-Vmu-coordinates}
\begin{aligned}
V_0
&=
U_0(\boldsymbol{\xi}_0)
+
U_0^*(\mathbf x_1)
-
\boldsymbol{\xi}_0^\top\mathbf x_1,\\
V_\mu
&=
U_{\mu-1}(\boldsymbol{\xi}_{\mu+1})
+
U_{\mu-1}^*(\mathbf x_\mu)
-
\boldsymbol{\xi}_{\mu+1}^\top\mathbf x_\mu,
\mu=1,\ldots,m-1,\\
V_m
&=
U_{m-1}^*(\mathbf x_m).
\end{aligned}
\end{equation}
By \eqref{eq:proof-x-xi} and Proposition~\ref{prop:potential-properties}, item~\ref{item:potential-gap}, \(U_{\mu-1}^*(\mathbf x_\mu)=\boldsymbol{\xi}_\mu^\top\mathbf x_\mu-U_{\mu-1}(\boldsymbol{\xi}_\mu)\). Moreover, item~\ref{item:potential-conjugate} of Proposition~\ref{prop:potential-properties} gives \(U_{m-1}^*(\mathbf x_m)=r_mU_{m-1}(\boldsymbol{\xi}_m)\). Hence, by the definition of \(H_\mu\),
\begin{equation}
\label{eq:Vmu-Hmu}
V_\mu(\mathbf x)
\equiv
H_\mu(\boldsymbol{\xi}(\mathbf x)),
\qquad
\mu=0,\ldots,m.
\end{equation}
For \(\mu=1,\ldots,m\), condition~2 of Definition~\ref{def:ahc} implies that \(\nabla U_{\mu-1}^*\) is continuously differentiable. Hence, \(\boldsymbol{\xi}_\mu(\cdot)\) is absolutely continuous and the chain rule, together with \eqref{eq:proof-x-xi} and \eqref{eq:proof-normalized-dynamics}, gives
\begin{equation}
\label{eq:proof-xi-derivative}
\begin{aligned}
\dot{\boldsymbol{\xi}}_\mu
&=
\nabla^2U_{\mu-1}^*(\mathbf x_\mu)
\dot{\mathbf x}_\mu\\
&=
-\widetilde{k}_\mu
\nabla^2U_{\mu-1}^*
\left(
\nabla U_{\mu-1}(\boldsymbol{\xi}_\mu)
\right)
\boldsymbol{\chi}_\mu(\boldsymbol{\xi};\mathbf s,\mathbf d),
\end{aligned}
\end{equation}
for \(
\mu=1,\ldots,m,\) and almost every \(t\). Since the functions in \eqref{eq:proof-Vmu-coordinates} are continuously differentiable in their arguments, \(V_\mu(\mathbf x(\cdot))\), \(\mu=0,\ldots,m\), are absolutely continuous. Therefore, \(V(\mathbf x(\cdot))\) is absolutely continuous and
\begin{equation}
\label{eq:proof-V-derivative}
\dot V
=
\sum_{\mu=0}^m
\beta_\mu H_\mu(\boldsymbol{\xi})^{\rho_\mu-1}\dot V_\mu
\end{equation}
for almost every \(t\). For \(V_0\), the chain rule, \eqref{eq:proof-x-xi}, and \eqref{eq:proof-normalized-dynamics} give
\begin{equation}
\label{eq:proof-V0-derivative}
\begin{aligned}
\dot V_0
&=
\left(
\nabla U_0(\boldsymbol{\xi}_0)-\mathbf x_1
\right)^\top
\dot{\mathbf x}_0
+
\left(
\nabla U_0^*(\mathbf x_1)-\boldsymbol{\xi}_0
\right)^\top
\dot{\mathbf x}_1\\
&=
\mathbf h_0(\boldsymbol{\xi})^\top\dot{\mathbf x}_0
+
(\boldsymbol{\xi}_1-\boldsymbol{\xi}_0)^\top
\dot{\mathbf x}_1\\
&=
-\widetilde{k}_0
\mathbf h_0(\boldsymbol{\xi})^\top
\mathbf Q
\mathbf h_0(\boldsymbol{\xi})
+
\theta_0(\boldsymbol{\xi};\mathbf s,\mathbf d;\widetilde{k}_1).
\end{aligned}
\end{equation}
For \(\mu=1,\ldots,m-1\), the chain rule applied to \eqref{eq:proof-Vmu-coordinates}, followed by adding and subtracting \(\nabla U_\mu(\boldsymbol{\xi}_\mu)\), gives
\begin{equation}
\label{eq:proof-Vmu-derivative}
\begin{aligned}
&\dot V_\mu
=\\&
\left(
\nabla U_{\mu-1}(\boldsymbol{\xi}_{\mu+1})-\mathbf x_\mu
\right)^\top
\dot{\boldsymbol{\xi}}_{\mu+1}
+
\left(
\nabla U_{\mu-1}^*(\mathbf x_\mu)-\boldsymbol{\xi}_{\mu+1}
\right)^\top
\dot{\mathbf x}_\mu\\
&=
\left(
\nabla U_{\mu-1}(\boldsymbol{\xi}_{\mu+1})
-
\nabla U_{\mu-1}(\boldsymbol{\xi}_\mu)
\right)^\top
\dot{\boldsymbol{\xi}}_{\mu+1}
+
(\boldsymbol{\xi}_\mu-\boldsymbol{\xi}_{\mu+1})^\top
\dot{\mathbf x}_\mu\\
&=
\theta_\mu(\boldsymbol{\xi};\mathbf s,\mathbf d;\widetilde{k}_{\mu+1})
-
\widetilde{k}_\mu
(\boldsymbol{\xi}_\mu-\boldsymbol{\xi}_{\mu+1})^\top
\mathbf h_\mu(\boldsymbol{\xi})=\\
&
\theta_\mu(\boldsymbol{\xi};\mathbf s,\mathbf d;\widetilde{k}_{\mu+1})
-
\widetilde{k}_\mu
(\boldsymbol{\xi}_\mu-\boldsymbol{\xi}_{\mu+1})^\top
\left(
\nabla U_\mu(\boldsymbol{\xi}_\mu)
-
\nabla U_\mu(\boldsymbol{\xi}_{\mu+1})
\right)\\
&-
\widetilde{k}_\mu
(\boldsymbol{\xi}_\mu-\boldsymbol{\xi}_{\mu+1})^\top
\left(
\nabla U_\mu(\boldsymbol{\xi}_0)
-
\nabla U_\mu(\boldsymbol{\xi}_\mu)
\right)\\
&=
-
\widetilde{k}_\mu
(\boldsymbol{\xi}_\mu-\boldsymbol{\xi}_{\mu+1})^\top
\left(
\nabla U_\mu(\boldsymbol{\xi}_\mu)
-
\nabla U_\mu(\boldsymbol{\xi}_{\mu+1})
\right)\\
&+
\widetilde{k}_\mu\psi_\mu(\boldsymbol{\xi})
+
\theta_\mu(\boldsymbol{\xi};\mathbf s,\mathbf d;\widetilde{k}_{\mu+1}).
\end{aligned}
\end{equation}
The first term is nonpositive by Proposition~\ref{prop:potential-properties}, item~\ref{item:potential-primal}, and \eqref{eq:strict-monotone}, while the rest of cross terms have no definite sign. Finally, from \(V_m=U_{m-1}^*(\mathbf x_m)\), \eqref{eq:proof-x-xi}, and \eqref{eq:proof-normalized-dynamics},
\begin{equation}
\label{eq:proof-Vm-derivative}
\dot V_m
=
-\widetilde{k}_m
\boldsymbol{\xi}_m^\top
\boldsymbol{\chi}_m(\boldsymbol{\xi};\mathbf s,\mathbf d)
=
\widetilde{k}_m
\psi_m(\boldsymbol{\xi};\mathbf s,\mathbf d).
\end{equation}
Combining \eqref{eq:proof-V-derivative}, \eqref{eq:proof-V0-derivative}, \eqref{eq:proof-Vmu-derivative}, \eqref{eq:proof-Vm-derivative}, and \eqref{eq:proof-Bnu} gives \(\dot V(\mathbf x(t))=-\Phi(\boldsymbol{\xi}(\mathbf x(t));\mathbf s(t),\mathbf d(t))\) for almost every \(t\). Since \(\Phi(\boldsymbol{\xi})\) is the minimum in \eqref{eq:proof-Phi}, this proves \eqref{eq:proof-Phi-upper}.
\end{proof}

\subsection{Recursive domination}

\begin{lemma}
\label{lem:recursive-estimate}
Let \(V\) and \(\Phi\) be defined by \eqref{eq:lyap} and \eqref{eq:proof-Phi}. For every \(k_m>1/c_{\mathcal S}\) and every \(\widetilde{k}_m>0\), there exist normalized gains \(\widetilde{k}_{m-1},\ldots,\widetilde{k}_0>0\) and a constant \(c>0\) such that
\begin{equation}
\label{eq:proof-Phi-V-bound}
\Phi(\boldsymbol{\xi}(\mathbf x))
\geq
cV(\mathbf x)^\alpha
\end{equation}
for every \(\mathbf x\in\mathcal X^{m+1}\), where
\begin{equation}
\label{eq:alpha}
\alpha
:=
1-\frac{1}{(m+1)d_V}
\in(0,1)
\end{equation}
and \(d_V\) is given in \eqref{eq:V-degree}. Consequently, every solution of \eqref{eq:normalized-system} with \(\gamma=0\) satisfies
\begin{equation}
\label{eq:main-lyapunov-estimate}
\dot V(\mathbf x(t))
\leq
-cV(\mathbf x(t))^\alpha
\end{equation}
for almost every \(t\).
\end{lemma}

\begin{proof}
We first establish the homogeneity of \(\Phi\). By Proposition~\ref{prop:potential-properties}, item~\ref{item:potential-conjugate}, each \(\boldsymbol{\xi}_\mu\) has degree one. Hence, \(\mathbf h_\mu\) has degree \(r_{\mu+1}\), while \(\boldsymbol{\chi}_m\) has degree zero. Moreover, item~\ref{item:potential-hessian} gives degree \(1-r_{\mu+1}\) for \(\nabla^2U_\mu^*(\nabla U_\mu(\boldsymbol{\xi}_{\mu+1}))\). Together with \eqref{eq:Vmu-Hmu} and \eqref{eq:V0-homogeneity}--\eqref{eq:V-degree}, these identities show that \(B_0\) and \(H_0^{\rho_0-1}\theta_0\) have degree
\begin{equation}
\label{eq:proof-q}
q
:=
d_V-\frac{1}{m+1}
>0.
\end{equation}
For \(\mu=1,\ldots,m-1\), \(B_\mu\), \(H_\mu^{\rho_\mu-1}\psi_\mu\), and \(H_\mu^{\rho_\mu-1}\theta_\mu\) have degree $q$. Finally, \(H_m^{\rho_m-1}\psi_m\) also has degree $q$. Together with \eqref{eq:proof-Phi} and the degree-zero homogeneity of \(\mathcal S\), these calculations show that \(\Phi\) is homogeneous of degree \(q\).

In the following, we repeatedly apply the domination argument of Proposition~\ref{prop:homogeneous-domination} in Appendix \ref{sec:auxiliary}. The domination is performed recursively through the continuous nonnegative terms \(B_\nu\) in \eqref{eq:proof-Phi-selection}, which are used to enforce positivity of \(\Phi\). We therefore first characterize the sets on which each \(B_\nu\) vanishes.

Set \(\mathcal C_{-1}:=\mathcal X^{m+1}\). By Lemma~\ref{lem:lyapunov-properties} and \eqref{eq:cmu:xi}, each \(\mathcal C_\mu\) is closed, contains the origin, and is invariant under \eqref{eq:weighted-dilation}. We claim that, for every \(\mathbf x\in\mathcal C_{\nu-1}\),
\begin{equation}
\label{eq:proof-Bnu-zero}
B_\nu(\boldsymbol{\xi}(\mathbf x))=0
\quad\Longleftrightarrow\quad
\mathbf x\in\mathcal C_\nu.
\end{equation}
To prove this, note that, for \(\mathbf x\in\mathcal C_{\nu-1}\), \eqref{eq:Vmu-Hmu} and the zero-set characterization of \(V_\nu\) give \(H_\nu(\boldsymbol{\xi}(\mathbf x))=0\) if and only if \(\boldsymbol{\xi}_\nu=\boldsymbol{\xi}_{\nu+1}\). For \(\nu=0\), positive definiteness of \(\mathbf Q+\mathbf Q^\top\) on \(\mathcal X\) and injectivity of \(\nabla U_0\) imply that \(\mathbf h_0^\top\mathbf Q\mathbf h_0=0\) if and only if \(\boldsymbol{\xi}_0=\boldsymbol{\xi}_1\). For \(\nu=1,\ldots,m-1\), Proposition~\ref{prop:potential-properties}, item~\ref{item:potential-primal} in Appendix \ref{sec:auxiliary}, and \eqref{eq:strict-monotone} give
\[
(\boldsymbol{\xi}_\nu-\boldsymbol{\xi}_{\nu+1})^\top
\left(
\nabla U_\nu(\boldsymbol{\xi}_\nu)
-
\nabla U_\nu(\boldsymbol{\xi}_{\nu+1})
\right)
>0
\]
whenever \(\boldsymbol{\xi}_\nu\neq\boldsymbol{\xi}_{\nu+1}\). Hence, in every case,
\[
B_\nu(\boldsymbol{\xi}(\mathbf x))=0
\quad\Longleftrightarrow\quad
\boldsymbol{\xi}_\nu=\boldsymbol{\xi}_{\nu+1}.
\]
Since \(\mathbf x\in\mathcal C_{\nu-1}\), \eqref{eq:cmu:xi} then gives \(\mathbf x\in\mathcal C_\nu\), proving \eqref{eq:proof-Bnu-zero}.

We now construct the gains by backward induction. First, we establish \(\Phi>0\) on \(\mathcal C_{m-1}\setminus\{\mathbf0\}\), which initializes the induction. Then, for \(\nu=m-1,\ldots,0\), assuming that \(\widetilde{k}_{\nu+1},\ldots,\widetilde{k}_m\) have been selected so that \(\Phi>0\) on \(\mathcal C_\nu\setminus\{\mathbf0\}\), we show that \(\widetilde{k}_\nu\) can be selected so that \(\Phi>0\) on \(\mathcal C_{\nu-1}\setminus\{\mathbf0\}\).

On \(\mathcal C_{m-1}\), \(\boldsymbol{\xi}_0=\cdots=\boldsymbol{\xi}_m\), and hence every term in \eqref{eq:proof-Phi-selection} except the last one vanishes. Thus, by \eqref{eq:proof-Phi},
\begin{equation}
\label{eq:proof-terminal-positivity}
\begin{aligned}
&\Phi(\boldsymbol{\xi}(\mathbf x))
=
\beta_m\widetilde{k}_m
H_m(\boldsymbol{\xi}(\mathbf x))^{\rho_m-1}
\min_{\substack{
\mathbf s\in\mathcal S(\boldsymbol{\xi}_0)\\
\mathbf d\in\mathcal D
}}
\boldsymbol{\xi}_0^\top
\left(
\mathbf s-\frac{1}{k_m}\mathbf d
\right)\\
&\geq
\beta_m\widetilde{k}_m
\left(
c_{\mathcal S}-\frac{1}{k_m}
\right)
H_m(\boldsymbol{\xi}(\mathbf x))^{\rho_m-1}
\|\boldsymbol{\xi}_0\|,
\quad
\mathbf x\in\mathcal C_{m-1},
\end{aligned}
\end{equation}
where the inequality follows from conditions~3 and~4 of Definition~\ref{def:ahc}. Since \(k_m>1/c_{\mathcal S}\), the coefficient is positive. Moreover, for every \(\mathbf x\in\mathcal C_{m-1}\setminus\{\mathbf0\}\), one has \(\boldsymbol{\xi}_0\neq\mathbf0\) and \(H_m(\boldsymbol{\xi})=r_mU_{m-1}(\boldsymbol{\xi}_0)>0\). Therefore, \(\Phi(\boldsymbol{\xi}(\mathbf x))>0\) for every \(\mathbf x\in\mathcal C_{m-1}\setminus\{\mathbf0\}\).

We now prove the induction step. Fix \(\nu\in\{0,\ldots,m-1\}\) and assume that \(\widetilde{k}_{\nu+1},\ldots,\widetilde{k}_m\) have been selected so that \(\Phi>0\) on \(\mathcal C_\nu\setminus\{\mathbf0\}\). Consider \(\mathbf x\in\mathcal C_{\nu-1}\). Then \(\boldsymbol{\xi}_0=\cdots=\boldsymbol{\xi}_\nu\), so all terms in \eqref{eq:proof-Phi-selection} involving \(\widetilde{k}_0,\ldots,\widetilde{k}_{\nu-1}\) vanish. Moreover, \(\psi_\nu=0\) when \(\nu\geq1\). Hence, on \(\mathcal C_{\nu-1}\), the dependence of \(\Phi\) on \(\widetilde{k}_\nu\) is isolated in the nonnegative term \(\beta_\nu\widetilde{k}_\nu B_\nu\), while all remaining terms depend only on \(\widetilde{k}_{\nu+1},\ldots,\widetilde{k}_m\). Since \(B_\nu\) is independent of \((\mathbf s,\mathbf d)\), taking the minimum in \eqref{eq:proof-Phi} gives
\begin{equation}
\label{eq:proof-recursive-decomposition}
\Phi(\boldsymbol{\xi}(\mathbf x))
=
\beta_\nu\widetilde{k}_\nu B_\nu(\boldsymbol{\xi}(\mathbf x))
-
A_\nu(\boldsymbol{\xi}(\mathbf x);\widetilde{k}_{\nu+1},\ldots,\widetilde{k}_m),
\end{equation}
\(\forall \mathbf{x}\in\mathcal C_{\nu-1},\) where
\begin{equation}
\label{eq:proof-Anu}
\begin{aligned}
&A_\nu(\boldsymbol{\xi};\widetilde{k}_{\nu+1},\ldots,\widetilde{k}_m)
:=\\&\max_{\substack{
\mathbf s\in\mathcal S(\boldsymbol{\xi}_0)\\
\mathbf d\in\mathcal D
}}
\Bigg\{
\beta_\nu H_\nu(\boldsymbol{\xi})^{\rho_\nu-1}
\theta_\nu(\boldsymbol{\xi};\mathbf s,\mathbf d;\widetilde{k}_{\nu+1})\\
&+
\sum_{\mu=\nu+1}^{m-1}
\beta_\mu\Bigg[
H_\mu(\boldsymbol{\xi})^{\rho_\mu-1}
\left(
\widetilde{k}_\mu\psi_\mu(\boldsymbol{\xi})
+
\theta_\mu(\boldsymbol{\xi};\mathbf s,\mathbf d;\widetilde{k}_{\mu+1})
\right)
\\&-
\widetilde{k}_\mu B_\mu(\boldsymbol{\xi})
\Bigg]+
\beta_mH_m(\boldsymbol{\xi})^{\rho_m-1}
\widetilde{k}_m
\psi_m(\boldsymbol{\xi};\mathbf s,\mathbf d)
\Bigg\}.
\end{aligned}
\end{equation}
The maximum is attained since \(\mathcal S(\boldsymbol{\xi}_0)\) and \(\mathcal D\) are compact. Moreover, continuity of the maximized function, together with upper semicontinuity and compactness of \(\mathcal S\) and compactness of \(\mathcal D\), shows that \(A_\nu\) is upper semicontinuous, while the preceding homogeneity calculation shows that it is homogeneous of degree \(q\). Since \(k_m\) and \(\beta_0,\ldots,\beta_m\) are fixed, their dependence is omitted from the notation.

By \eqref{eq:proof-Bnu-zero}, if \(\mathbf x\in\mathcal C_{\nu-1}\setminus\{\mathbf0\}\) and \(B_\nu(\boldsymbol{\xi}(\mathbf x))=0\), then \(\mathbf x\in\mathcal C_\nu\setminus\{\mathbf0\}\). Hence, by the induction hypothesis and \eqref{eq:proof-recursive-decomposition},
\[
\begin{aligned}
&A_\nu(\boldsymbol{\xi}(\mathbf x);\widetilde{k}_{\nu+1},\ldots,\widetilde{k}_m)
=
-\Phi(\boldsymbol{\xi}(\mathbf x))
<0,\\&
\forall\,\mathbf x\in\mathcal C_{\nu-1}\setminus\{\mathbf0\}
\text{ with } B_\nu(\boldsymbol{\xi}(\mathbf x))=0.
\end{aligned}
\]

This is the hypothesis of Proposition~\ref{prop:homogeneous-domination}, applied on \(\mathcal C_{\nu-1}\) with \(a(\mathbf x):=A_\nu(\boldsymbol{\xi}(\mathbf x);\widetilde{k}_{\nu+1},\ldots,\widetilde{k}_m)\) and \(b(\mathbf x):=\beta_\nu B_\nu(\boldsymbol{\xi}(\mathbf x))\), which yields the finite threshold \(\widetilde{k}_\nu^*=\omega_\nu(\widetilde{k}_{\nu+1},\ldots,\widetilde{k}_m)\), where
\begin{equation}
\label{eq:proof-knu-star}
\begin{aligned}
&\omega_\nu(\widetilde{k}_{\nu+1},\ldots,\widetilde{k}_m)
:=\\&
\max\left\{
0,\,
\sup_{\substack{
\mathbf x\in\mathcal C_{\nu-1},\\ \varrho(\mathbf x)=1\\
B_\nu(\boldsymbol{\xi}(\mathbf x))>0
}}
\frac{
A_\nu(\boldsymbol{\xi}(\mathbf x);\widetilde{k}_{\nu+1},\ldots,\widetilde{k}_m)
}{
\beta_\nu B_\nu(\boldsymbol{\xi}(\mathbf x))
}
\right\}.
\end{aligned}
\end{equation}
Therefore, from \eqref{eq:proof-recursive-decomposition}, for every \(\widetilde{k}_\nu>\widetilde{k}_\nu^*\), there exists \(c_\nu>0\) such that
\begin{equation}
\label{eq:proof-recursive-positive}
\Phi(\boldsymbol{\xi}(\mathbf x))
\geq
c_\nu\varrho(\mathbf x)^q,
\qquad
\mathbf x\in\mathcal C_{\nu-1}.
\end{equation}
Thus, selecting \(\widetilde{k}_\nu>\widetilde{k}_\nu^*\) establishes the induction claim on \(\mathcal C_{\nu-1}\). Repeating this step for \(\nu=m-1,\ldots,0\), and using \(\mathcal C_{-1}=\mathcal X^{m+1}\), yields \(c_0>0\) such that
\begin{equation}
\label{eq:proof-global-Phi}
\Phi(\boldsymbol{\xi}(\mathbf x))
\geq
c_0\varrho(\mathbf x)^q
\end{equation}
for every \(\mathbf x\in\mathcal X^{m+1}\).

It remains to express \eqref{eq:proof-global-Phi} in terms of \(V\). Since \(V\) is continuous and homogeneous of degree \(d_V\), there exists \(\overline c_V>0\) such that \(V(\mathbf x)\leq\overline c_V\varrho(\mathbf x)^{d_V}\) for every \(\mathbf x\). Therefore,
\[
\varrho(\mathbf x)^q
\geq
\overline c_V^{-q/d_V}V(\mathbf x)^{q/d_V}.
\]
By \eqref{eq:proof-q}, \(q/d_V=1-1/((m+1)d_V)=\alpha\). Combined with \eqref{eq:proof-global-Phi} proves \eqref{eq:proof-Phi-V-bound} with \(c:=c_0\overline c_V^{-\alpha}\). Lemma~\ref{lem:lyapunov-upper-bound} gives \eqref{eq:main-lyapunov-estimate}.
\end{proof}

\subsection{Proof of Theorem~\ref{thm:main}}

\begin{proof}
Choose \(k_m>1/c_{\mathcal S}\) and \(\widetilde{k}_m>0\). By Lemma~\ref{lem:recursive-estimate}, there exist \(\widetilde{k}_{m-1},\ldots,\widetilde{k}_0>0\) and \(c>0\) such that \eqref{eq:main-lyapunov-estimate} holds for \(\gamma=0\). Set \(\kappa:=k_m/\prod_{\mu=0}^m\widetilde{k}_\mu\), \(k_0:=\widetilde{k}_0\kappa^{1/(m+1)}\), and \(k_\mu:=\widetilde{k}_\mu k_{\mu-1}\kappa^{1/(m+1)}\) for \(\mu=1,\ldots,m-1\). Then \(\kappa>0\), and these gains reproduce the normalized gains of Lemma~\ref{lem:normalization}. Hence, \eqref{eq:main-lyapunov-estimate} holds for \eqref{eq:normalized-system} when \(\gamma=0\).

Consider now \(\gamma\geq0\), let \(\mathbf x(t)\) be any solution of \eqref{eq:normalized-system}, and define \[
\begin{aligned}
\mathbf z(t)&:=\delta(\mathbf x(t);\exp(\gamma t)),\\
\tau(t)&:=\int_0^t \exp(\gamma s/(m+1))\,\mathrm ds.
\end{aligned}
\] 

Differentiating \(\mathbf z_\mu(t)=\exp(\gamma r_\mu t)\mathbf x_\mu(t)\) gives \(\dot{\mathbf z}_\mu=\gamma r_\mu\mathbf z_\mu+\exp(\gamma r_\mu t)\dot{\mathbf x}_\mu\). The main idea is that substituting \eqref{eq:normalized-system}, the term \(\gamma r_\mu\mathbf z_\mu\) obtained from the product rule before cancels with \(-\gamma r_\mu\exp(\gamma r_\mu t)\mathbf x_\mu=-\gamma r_\mu\mathbf z_\mu\) in the dynamics. Since \(r_{\mu+1}=r_\mu-1/(m+1)\), \(\nabla U_\mu\) is homogeneous of degree \(r_{\mu+1}\), and \(\mathcal S\) is homogeneous of degree zero, the remaining terms have the common factor \(\exp(\gamma t/(m+1))=\mathrm d\tau/\mathrm dt\). Hence,
\begin{equation}
\begin{aligned}
\frac{\mathrm d\mathbf z_0}{\mathrm d\tau}
&=
-\widetilde{k}_0\mathbf Q
\left(
\nabla U_0(\mathbf z_0)-\mathbf z_1
\right),\\
\frac{\mathrm d\mathbf z_\mu}{\mathrm d\tau}
&=
-\widetilde{k}_\mu
\left(
\nabla U_\mu(\mathbf z_0)-\mathbf z_{\mu+1}
\right),
\quad \mu=1,\ldots,m-1,\\
\frac{\mathrm d\mathbf z_m}{\mathrm d\tau}
&\in
-\widetilde{k}_m
\left(
\mathcal S(\mathbf z_0)-\frac{1}{k_m}\mathcal D
\right).
\end{aligned}
\end{equation}

Since \(\mathrm d\tau/\mathrm dt>0\), the reparameterized curve \(\mathbf z(t(\tau))\) satisfies \eqref{eq:normalized-system} with \(\gamma=0\), and \eqref{eq:main-lyapunov-estimate} gives \(\mathrm dV(\mathbf z)/\mathrm d\tau\leq-cV(\mathbf z)^\alpha\) almost everywhere. Using \(V(\mathbf z(t))=\exp(d_V\gamma t)V(\mathbf x(t))\), which follows from \eqref{eq:V-homogeneity}, the chain rule gives
\begin{equation}
\label{eq:damped-lyapunov-estimate}
\begin{aligned}
&\frac{\mathrm dV(\mathbf x(t))}{\mathrm dt}
=\\&
-\gamma d_V\exp(-d_V\gamma t)V(\mathbf z(t))
+
\exp(-d_V\gamma t)
\frac{\mathrm dV(\mathbf z(t))}{\mathrm dt}\\
&=
-\gamma d_VV(\mathbf x(t))
+
\exp(-d_V\gamma t)
\frac{\mathrm dV(\mathbf z)}{\mathrm d\tau}
\frac{\mathrm d\tau}{\mathrm dt}\\
&\leq
-
c\exp(-d_V\gamma t)
V(\mathbf z(t))^\alpha
\exp\left(\frac{\gamma t}{m+1}\right)\\
&=
-
c\exp\left(
\gamma t\left(
-d_V+\frac{1}{m+1}+d_V\alpha
\right)
\right)
V(\mathbf x(t))^\alpha\\
&=
-cV(\mathbf x(t))^\alpha,
\end{aligned}
\end{equation}
for almost every \(t\), where \(d_V\alpha=d_V-1/(m+1)\) follows from \eqref{eq:alpha}. 

By Lemma~\ref{lem:lyapunov-properties}, \(V\) is positive definite, radially unbounded, and continuously differentiable. Hence, \(V(\mathbf x(\cdot))\) is absolutely continuous, and \eqref{eq:damped-lyapunov-estimate} implies that every solution remains in the compact sublevel set determined by its initial condition. Since the right-hand side of \eqref{eq:normalized-system} is upper semicontinuous, locally bounded, and has nonempty compact convex values, every maximal solution is forward complete. Moreover, since \(\alpha\in(0,1)\), integration of \eqref{eq:damped-lyapunov-estimate} shows that every solution reaches the origin in finite time, with \(T\leq V(\mathbf x(0))^{1-\alpha}/(c(1-\alpha))\). Equation~\eqref{eq:damped-lyapunov-estimate} also implies that every solution starting at or reaching the origin remains there, while positive definiteness of \(V\) and its nonincrease imply Lyapunov stability. Therefore, the origin of \eqref{eq:normalized-system} is globally finite-time stable. Lemma~\ref{lem:normalization} gives the same conclusion for \eqref{eq:system}.
\end{proof}

\section{Gain proposal}
\label{sec:gain-selection}

Although Theorem~\ref{thm:main} is existential with respect to the gains, its proof yields the recursive gain-proposal procedure summarized in Algorithm~\ref{alg:gain-selection}. At each step, the required gain threshold is characterized by a finite-dimensional optimization problem involving quantities determined by the AHC data and the selected Lyapunov coefficients \(\beta_0,\ldots,\beta_m\).

\begin{algorithm}
\caption{Recursive gain-proposal}
\label{alg:gain-selection}
\begin{algorithmic}[1]

\STATE Choose \(\beta_0,\ldots,\beta_m>0\), \(k_m>1/c_{\mathcal S}\), and \(\widetilde{k}_m>0\).

\FOR{\(\nu=m-1,m-2,\ldots,0\)}

    \STATE Determine
    \(
    \widetilde{k}_\nu^*
=
\omega_\nu(\widetilde{k}_{\nu+1},\ldots,\widetilde{k}_m)
    \)
    from \eqref{eq:gain-selection-program}.

    \STATE Choose \(\widetilde{k}_\nu>\widetilde{k}_\nu^*\).

\ENDFOR

\STATE Set
\(
\kappa
:=
{k_m}
\left(\prod_{\mu=0}^m\widetilde{k}_\mu\right)^{-1}.
\)

\STATE Set
\(
k_0:=\widetilde{k}_0\kappa^{1/(m+1)}.
\)

\FOR{\(\mu=1,2,\ldots,m-1\)}

    \STATE Set
    \(
    k_\mu
    :=
    \widetilde{k}_\mu k_{\mu-1}\kappa^{1/(m+1)}.
    \)

\ENDFOR

\end{algorithmic}
\end{algorithm}

\begin{proposition}
\label{prop:gain-selection}
Assume that the conditions of Theorem~\ref{thm:main} hold. At each step \(\nu=m-1,\ldots,0\) of Algorithm~\ref{alg:gain-selection}, with \(\widetilde{k}_{\nu+1},\ldots,\widetilde{k}_m\) fixed according to the preceding steps, the threshold \(\omega_\nu(\widetilde{k}_{\nu+1},\ldots,\widetilde{k}_m)\) in \eqref{eq:proof-knu-star} is equivalently represented by
\begin{equation}
\label{eq:gain-selection-program}
\begin{aligned}
&\omega_\nu(\widetilde{k}_{\nu+1},\ldots,\widetilde{k}_m)
=
\max\left\{
0,\,
\underset{k,\boldsymbol{\xi}}{\operatorname{sup}}
\quad
k
\right\}\\
&\operatorname{subject\ to}
\quad
A_\nu(\boldsymbol{\xi};
\widetilde{k}_{\nu+1},\ldots,\widetilde{k}_m)
-
k\beta_\nu B_\nu(\boldsymbol{\xi})
\geq0,
\\&\boldsymbol{\xi}\in\mathcal X^{m+1},
\qquad
\|\boldsymbol{\xi}\|=1,
\qquad
\boldsymbol{\xi}_0=\cdots=\boldsymbol{\xi}_\nu.
\end{aligned}
\end{equation}
Consequently, if the gains are selected according to Algorithm~\ref{alg:gain-selection}, then the conclusion of Theorem~\ref{thm:main} holds.
\end{proposition}

\begin{proof}
By \eqref{eq:cmu:xi}, the constraint \(\mathbf x\in\mathcal C_{\nu-1}\) is equivalent to \(\boldsymbol{\xi}_0=\cdots=\boldsymbol{\xi}_\nu\). Moreover, by \eqref{eq:proof-xi} and Proposition~\ref{prop:potential-properties}, item~\ref{item:potential-conjugate}, the transformation \(\boldsymbol{\xi}(\mathbf x)\) is bijective and satisfies \(\boldsymbol{\xi}(\delta(\mathbf x;\lambda))=\lambda\boldsymbol{\xi}(\mathbf x)\). Hence, every nonzero homogeneous ray in \(\mathcal C_{\nu-1}\) has a unique representative satisfying \(\|\boldsymbol{\xi}\|=1\). Since \(A_\nu\) and \(\beta_\nu B_\nu\) are homogeneous of the same degree, their quotient is invariant along these rays wherever \(B_\nu>0\). If \(B_\nu(\boldsymbol{\xi})=0\), then \eqref{eq:proof-Bnu-zero} and the recursive construction give \(A_\nu(\boldsymbol{\xi};\widetilde{k}_{\nu+1},\ldots,\widetilde{k}_m)<0\), so the inequality constraint in \eqref{eq:gain-selection-program} cannot hold. Hence, every feasible point satisfies \(B_\nu(\boldsymbol{\xi})>0\), and the inequality constraint is equivalent to
\(k\leq A_\nu(\boldsymbol{\xi};\widetilde{k}_{\nu+1},\ldots,\widetilde{k}_m)/(\beta_\nu B_\nu(\boldsymbol{\xi}))\). Therefore, taking the supremum over \(k\) and \(\boldsymbol{\xi}\) yields exactly the supremum in \eqref{eq:proof-knu-star}, and taking the maximum with zero gives \(\omega_\nu(\widetilde{k}_{\nu+1},\ldots,\widetilde{k}_m)\). Algorithm~\ref{alg:gain-selection} consequently reproduces the gain construction in the proof of Theorem~\ref{thm:main}, and the conclusion follows.
\end{proof}

\begin{remark}
Algorithm~\ref{alg:gain-selection} gives a sequential gain-proposal procedure, although the thresholds \(\omega_\nu\) generally require approximating \eqref{eq:gain-selection-program} numerically, as in other established HOSM gain-design procedures \cite{cruz2018,calmbach2026}. Once \(\beta_0,\ldots,\beta_m\) are fixed, all quantities entering this program are determined by the AHC data through \(U_\mu\), \(\nabla U_\mu\), and the Hessian actions of \(U_\mu^*\) appearing in \eqref{eq:proof-auxiliary-functions}. In particular, for the power potentials in Proposition~\ref{prop:power-potential}, these Hessian actions can be computed through \eqref{eq:power-potential-hessian-characterization} without evaluating the conjugates \(U_\mu^*\) explicitly.
\end{remark}

\section{Numerical example}
\label{sec:numerical-example}

We illustrate the gain-proposal procedure with the affine formation observer of Section~\ref{subsec:affine}. We consider the planar nominal configuration shown in Figure~\ref{fig:affine-configuration}, taken from \cite{zhao2018}, with \(N=7\), \(N_L=3\), \(N_F=4\), \(d=2\), and observer order \(m=2\).

\begin{figure}[t]
\centering
\includegraphics[width=0.78\columnwidth]{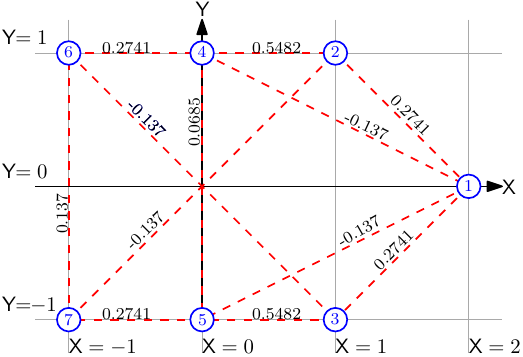}
\caption{\emph{Nominal configuration and relative stress parameters used in the numerical example, adapted from \cite{zhao2018}. Nodes \(\{1,2,3\}\) are leaders and nodes \(\{4,5,6,7\}\) are followers.}}
\label{fig:affine-configuration}
\end{figure}

The common scale of the stress parameters is arbitrary, since it does not change either the affine formation space or the follower reference. We therefore use this freedom to normalize \(c_{\mathcal S}\). Let \(\mathbf\Omega^\star\) denote the stress matrix associated with the parameters shown in Figure~\ref{fig:affine-configuration}. We choose \(\alpha_\Omega=1/\lambda_{\min}(\mathbf\Omega_{FF}^\star)=2.9142\) and use the scaled stress matrix \(\mathbf\Omega=\alpha_\Omega\mathbf\Omega^\star\). Then \(c_{\mathcal S}=\lambda_{\min}(\mathbf\Omega_{FF})=1\), so the terminal gain condition \eqref{eq:km} reduces to \(k_2>1\). This scaling leaves \(\operatorname{Null}(\mathbf\Omega\otimes\mathbf I_d)\) unchanged so the prescribed affine formation and follower reference are unaffected.

We set \(\beta_0=\beta_1=\beta_2=1\), choose \(k_2=2>1/c_{\mathcal S}\) and \(\widetilde{k}_2=1\), and approximate the optimization problems in \eqref{eq:gain-selection-program} using the \texttt{\small differential\_evolution} routine from \texttt{\small scipy}, with \texttt{\small popsize}$=10$, \texttt{\small maxiter}$=100$, tolerance \(10^{-13}\), and ten independent runs. In line with related numerical HOSM gain-design procedures \cite{cruz2018,calmbach2026}, the searches evaluate the nonconvex homogeneous optimization problems on their compact normalized domains. Accordingly, the resulting values are used only to propose candidate gains, whose behavior is subsequently evaluated in simulation. The resulting threshold estimates are \(\omega_1=5.81\) and \(\omega_0=23.09\). Using these estimates as practical gain proposals, we select
\(
(\widetilde{k}_0,\widetilde{k}_1,\widetilde{k}_2)
=
(24,6,1),
\)
Algorithm~\ref{alg:gain-selection} then yields \(\kappa=1/72\) and
\(
(k_0,k_1,k_2)=(5.77,8.32,2).
\)
We next simulate \eqref{eq:affine:observer} with time-varying leader trajectories
\begin{equation}
\begin{aligned}
\mathbf p_1(t)&=\operatorname{col}(\sin t,\cos t),\\
\mathbf p_2(t)&=\operatorname{col}(1+0.8\sin(1.2t),0.8\cos(1.2t)),\\
\mathbf p_3(t)&=\operatorname{col}(0.7\cos(1.4t),1+0.7\sin(1.4t)).
\end{aligned}
\end{equation}
We select \(L=66\), which satisfies Assumption~\ref{ass:affine} for these trajectories, and initialize all follower observer states at zero. The observer dynamics \eqref{eq:affine:observer} are integrated using the forward Euler method with step size \(\Delta t=10^{-4}\). Figure~\ref{fig:affine-numerical} shows the horizontal $\mathsf{X}$ components of the observer states and corresponding follower references for \(\mu=0,1,2\), together with the Lyapunov function \(V\) and the numerically evaluated \(-\dot V\). The observer states reach the corresponding references in finite time, while \(V\) decreases to zero and \(-\dot V\) remains nonnegative up to numerical precision, consistently with the Lyapunov analysis.

\begin{figure}[t]
\centering
\includegraphics[width=0.9\columnwidth]{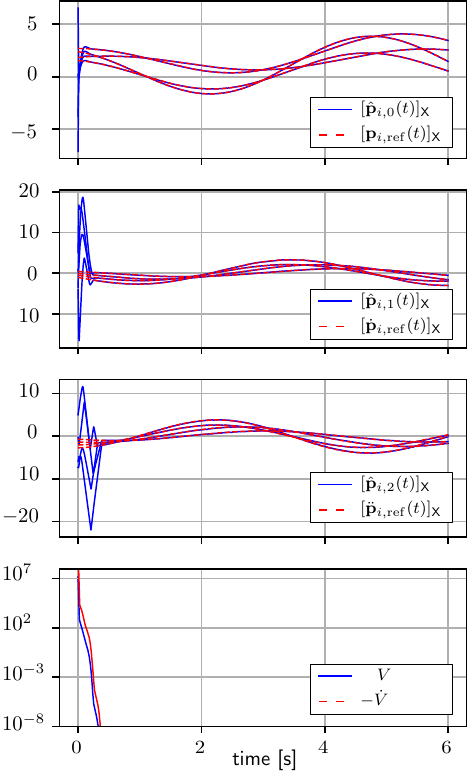}
\caption{\emph{Affine formation observer with \(m=2\). From top to bottom: horizontal components of the observer states (blue) and corresponding references (red dashed) for all followers and \(\mu=0,1,2\), followed by the Lyapunov function \(V\) (blue) and the numerically evaluated \(-\dot V\) (red), shown on a logarithmic scale.}}
\label{fig:affine-numerical}
\end{figure}

\section{Conclusions}

This work introduced abstract homogeneous chains as a unified framework for arbitrary-order sliding-mode algorithms in multi-agent systems. Combining homogeneity with convex analysis, we constructed an arbitrary-order homogeneous Lyapunov function, established global finite-time stability, and obtained recursive sufficient gain conditions through finite-dimensional optimization. The framework provides a numerical optimization-based arbitrary-order gain-proposal procedure for EDCHO, strengthens REDCHO convergence from local to global, provides arbitrary-order tuning rules for leader--follower distributed differentiation, and yields a new arbitrary-order observer for multi-leader affine formation tracking.

The framework also has two main limitations. Since the potentials \(U_\mu\) and the matrix \(\mathbf Q\) depend on the interaction graph, the Lyapunov construction is not common to different graphs and therefore does not directly cover switching networks. Moreover, although gain proposal is finite-dimensional, the required thresholds are generally not available in closed form and their certified computation may require non-trivial application-specific representations and optimization tools. These two directions are therefore considered for future work.

\appendix

\subsection{Auxiliary results}
\label{sec:auxiliary}
\begin{proposition}
\label{prop:potential-properties}
Let \(U:\mathcal{X}\to\R_{\geq0}\) satisfy condition~2 of Definition~\ref{def:ahc}, with homogeneity degree \(p\in(1,2)\), and let \(p^*:=p/(p-1)\). Then, the following properties hold.
\begin{enumerate}
\item\label{item:potential-primal} The function \(U\) is strictly convex, and the map \(\nabla U:\mathcal{X}\to\mathcal{X}\) is bijective. Moreover, \(\nabla U(\mathbf0)=\mathbf0\), \(\nabla U(\lambda\mathbf{x})=\lambda^{p-1}\nabla U(\mathbf{x})\) for every \(\mathbf{x}\in\mathcal{X}\) and \(\lambda>0\), and \(\mathbf{x}^\top\nabla U(\mathbf{x})=pU(\mathbf{x})\) for every \(\mathbf{x}\in\mathcal X\).

\item\label{item:potential-conjugate} The function \(U^*\) is positive definite and homogeneous of degree \(p^*\). Its gradient satisfies \(\nabla U^*=(\nabla U)^{-1}\), \(\nabla U^*(\mathbf0)=\mathbf0\), and \(\nabla U^*(\lambda\mathbf y)=\lambda^{1/(p-1)}\nabla U^*(\mathbf y)\) for every \(\mathbf y\in\mathcal X\) and \(\lambda>0\). Moreover, \(U^*(\nabla U(\mathbf x))=(p-1)U(\mathbf x)\) for every \(\mathbf x\in\mathcal X\).

\item\label{item:potential-hessian} The Hessian of \(U^*\) satisfies
\(
\nabla^2U^*(\lambda^{p-1}\mathbf y)
=
\lambda^{2-p}\nabla^2U^*(\mathbf y)
\)
for every \(\mathbf y\in\mathcal X\) and \(\lambda>0\).

\item\label{item:potential-gap} The Fenchel--Young gap \(F_U(\mathbf{x},\mathbf{y})\) is continuously differentiable and nonnegative. It vanishes if and only if \(\mathbf{y}=\nabla U(\mathbf{x})\), or equivalently, \(\mathbf{x}=\nabla U^*(\mathbf{y})\). Moreover, \(F_U(\lambda\mathbf{x},\lambda^{p-1}\mathbf{y})=\lambda^pF_U(\mathbf{x},\mathbf{y})\) for every \(\lambda>0\).
\end{enumerate}
\end{proposition}

\begin{proof}
For item~\ref{item:potential-primal}, continuity, positive definiteness, and homogeneity of \(U\) give \(c:=\min_{\|\mathbf{x}\|=1}U(\mathbf{x})>0\), hence \(U(\mathbf{x})\geq c\|\mathbf{x}\|^p\). Since \(U^*\) is differentiable on \(\mathcal X\), \cite[Theorem~26.3]{rockafellar1970} implies that \(U\) is strictly convex. Hence, \eqref{eq:strict-monotone} implies that \(\nabla U\) is injective. To prove surjectivity, fix \(\mathbf y\in\mathcal X\) and define \(\varphi_{\mathbf y}(\mathbf x):=U(\mathbf x)-\mathbf x^\top\mathbf y\). Since \(\varphi_{\mathbf y}(\mathbf x)\geq c\|\mathbf x\|^p-\|\mathbf x\|\|\mathbf y\|\) and \(p>1\), \(\varphi_{\mathbf y}\) is coercive and attains a minimum at some \(\mathbf x_{\mathbf y}\), where \(\nabla U(\mathbf x_{\mathbf y})=\mathbf y\). Thus, \(\nabla U\) is bijective. Since \(\mathbf0\) minimizes \(U\), \(\nabla U(\mathbf0)=\mathbf0\). Differentiating \(U(\lambda\mathbf x)=\lambda^pU(\mathbf x)\) with respect to \(\mathbf x\) gives \(\nabla U(\lambda\mathbf x)=\lambda^{p-1}\nabla U(\mathbf x)\), while differentiating with respect to \(\lambda\) at \(\lambda=1\) gives \(\mathbf x^\top\nabla U(\mathbf x)=pU(\mathbf x)\).

For item~\ref{item:potential-conjugate}, \cite[Theorem~16.23]{bauschke2011} gives \(\mathbf y=\nabla U(\mathbf x)\) if and only if \(\mathbf x=\nabla U^*(\mathbf y)\). Since \(\nabla U\) is bijective by item~\ref{item:potential-primal}, it follows that \(\nabla U^*=(\nabla U)^{-1}\). Also, \(U^*(\mathbf0)=\sup_{\mathbf x}\{-U(\mathbf x)\}=0\). For \(\mathbf y\neq\mathbf0\), let \(\mathbf x:=\nabla U^*(\mathbf y)\). Then \(\mathbf x\neq\mathbf0\), \(\mathbf y=\nabla U(\mathbf x)\), and \(U^*(\mathbf y)=\mathbf x^\top\mathbf y-U(\mathbf x)=(p-1)U(\mathbf x)>0\), so \(U^*\) is positive definite. For \(\lambda>0\), substituting \(\mathbf x=\lambda^{1/(p-1)}\mathbf z\) in the definition of \(U^*\) gives \(U^*(\lambda\mathbf y)=\lambda^{p/(p-1)}U^*(\mathbf y)=\lambda^{p^*}U^*(\mathbf y)\). Moreover, \(\nabla U(\lambda^{1/(p-1)}\mathbf x)=\lambda\nabla U(\mathbf x)\), and applying \((\nabla U)^{-1}=\nabla U^*\) gives \(\nabla U^*(\lambda\mathbf y)=\lambda^{1/(p-1)}\nabla U^*(\mathbf y)\), including \(\nabla U^*(\mathbf0)=\mathbf0\).

For item~\ref{item:potential-hessian}, item~\ref{item:potential-conjugate} gives \(\nabla U^*(\lambda^{p-1}\mathbf y)=\lambda\nabla U^*(\mathbf y)\). Differentiating with respect to \(\mathbf y\) gives \(\nabla^2U^*(\lambda^{p-1}\mathbf y)=\lambda^{2-p}\nabla^2U^*(\mathbf y)\).

For item~\ref{item:potential-gap}, the Fenchel--Young inequality \cite[Proposition~13.13]{bauschke2011} gives \(F_U(\mathbf{x},\mathbf{y})\geq0\), with equality if and only if \(\mathbf y=\nabla U(\mathbf x)\), equivalently, \(\mathbf x=\nabla U^*(\mathbf y)\), by \cite[Theorem~16.23]{bauschke2011}. Since \(U\) and \(U^*\) are continuously differentiable, so is \(F_U\). Finally, their homogeneity and \((p-1)p^*=p\) give \(F_U(\lambda\mathbf{x},\lambda^{p-1}\mathbf{y})=\lambda^pF_U(\mathbf{x},\mathbf{y})\).
\end{proof}

Homogeneous domination arguments of the following type are classical \cite{andrieu2008}. The version below allows a relaxed upper semi-continuity condition required in this work.
\begin{proposition}
\label{prop:homogeneous-domination}
Let \(\mathcal C\subseteq\mathcal X^{m+1}\) be nonempty, closed, and invariant under \eqref{eq:weighted-dilation}. Let \(a:\mathcal C\to\R\) be upper semicontinuous and let \(b:\mathcal C\to\R_{\geq0}\) be continuous. Suppose that \(a\) and \(b\) are homogeneous of the same degree \(q>0\), and that \(a(\mathbf x)<0\) whenever \(\mathbf x\in\mathcal C\setminus\{\mathbf0\}\) and \(b(\mathbf x)=0\). Define
\[
k^*
:=
\max\left\{
0,\,
\sup_{\substack{\mathbf x\in\mathcal C,\ \varrho(\mathbf x)=1\\
b(\mathbf x)>0}}
\frac{a(\mathbf x)}{b(\mathbf x)}
\right\},
\]
where the supremum over the empty set is understood as \(-\infty\). Then, \(k^*<\infty\), and for any \(k>k^*\) there exists \(c_k>0\) such that
\[
kb(\mathbf x)-a(\mathbf x)
\geq
c_k\varrho(\mathbf x)^q,
\qquad
\forall\,\mathbf x\in\mathcal C.
\]
\end{proposition}

\begin{proof}
If \(\mathcal C=\{\mathbf0\}\), the result is immediate. Otherwise, let
\(
\mathcal K:=\{\mathbf x\in\mathcal C:\varrho(\mathbf x)=1\},
\)
which is nonempty and compact. Let \(\mathcal Z:=\{\mathbf x\in\mathcal K:b(\mathbf x)=0\}\). Since \(b\) is continuous, \(\mathcal Z\) is compact. If \(\mathcal Z\neq\emptyset\), the hypothesis and upper semicontinuity of \(a\) give \(\max_{\mathbf x\in\mathcal Z}a(\mathbf x)<0\), hence \(a<0\) on a neighborhood of \(\mathcal Z\) in \(\mathcal K\). If \(\mathcal Z=\emptyset\), take this neighborhood to be empty. On its compact complement, \(b\) is bounded away from zero and \(a\) is bounded above. Hence,
\(
\sup_{\mathbf x\in\mathcal K,\ b(\mathbf x)>0}a(\mathbf x)/b(\mathbf x)<\infty,
\)
so \(k^*<\infty\). Fix \(k>k^*\). By the definition of \(k^*\) and the hypothesis on the zero set of \(b\), \(kb-a>0\) on \(\mathcal K\). Since \(kb-a\) is lower semicontinuous,
\(
c_k:=\min_{\mathbf x\in\mathcal K}(kb(\mathbf x)-a(\mathbf x))>0.
\)
For every \(\mathbf x\in\mathcal C\setminus\{\mathbf0\}\), let
\(
\widehat{\mathbf x}:=\delta(\mathbf x;1/\varrho(\mathbf x))\in\mathcal K.
\)
Homogeneity gives
\(
kb(\mathbf x)-a(\mathbf x)
=
\varrho(\mathbf x)^q
\left(
kb(\widehat{\mathbf x})-a(\widehat{\mathbf x})
\right)
\geq
c_k\varrho(\mathbf x)^q.
\)
The result at \(\mathbf x=\mathbf0\) follows from homogeneity.
\end{proof}

\begin{proposition}
\label{prop:power-potential}
Let \(\mathcal X\subseteq\R^n\) be a linear subspace, let \(\mathbf M=\operatorname{col}(\mathbf M_1,\ldots,\mathbf M_s)\), with \(\mathbf M_\ell\in\R^{d_\ell\times n}\), be injective on \(\mathcal X\), and let \(r\in(0,1)\). Define
\[
U(\mathbf x)
:=
\frac{1}{1+r}
\sum_{\ell=1}^{s}
\|\mathbf M_\ell\mathbf x\|^{1+r},
\qquad
\mathbf x\in\mathcal X.
\]
Then, \(U\) is continuously differentiable, convex, positive definite, and homogeneous of degree \(1+r\), and its conjugate \(U^*\) is twice continuously differentiable.
\end{proposition}

\begin{proof}
Continuous differentiability, convexity, and homogeneity follow from the corresponding properties of \(\mathbf z\mapsto\|\mathbf z\|^{1+r}/(1+r)\), while positive definiteness follows from the injectivity of \(\mathbf M\) on \(\mathcal X\). Hence, \(c:=\min_{\|\mathbf x\|=1}U(\mathbf x)>0\), and homogeneity gives \(U(\mathbf x)\geq c\|\mathbf x\|^{1+r}\). Moreover, since \(\mathbf M\) is injective on \(\mathcal X\), there exists \(\sigma_{\mathbf M}>0\) such that \(\|\mathbf M\mathbf x\|\geq\sigma_{\mathbf M}\|\mathbf x\|\) for every \(\mathbf x\in\mathcal X\).

The rest of the proof proceeds through four claims. Here, \(J[\mathbf f](\mathbf x)\) denotes the Jacobian of a map \(\mathbf f\) at \(\mathbf x\). First, we show that \(U^*\) is finite and differentiable and that \(\nabla U^*\) is locally Lipschitz. Second, we show that the map \(\mathbf T_{\mathbf x}\) in \eqref{eq:power-potential-Tx} is well defined and linear. Third, we prove that \(J[\nabla U^*](\mathbf y)=\mathbf T_{\mathbf x}\) for \(\mathbf y=\nabla U(\mathbf x)\). Fourth, we show that \(\mathbf T_{\mathbf x}\) depends continuously on \(\mathbf x\). The last two claims imply that \(U^*\) is twice continuously differentiable.

\emph{Claim 1: \(U^*\) is finite and differentiable, and \(\nabla U^*\) is locally Lipschitz.} Define \(\boldsymbol\phi(\mathbf z):=\|\mathbf z\|^{r-1}\mathbf z\), with \(\boldsymbol\phi(\mathbf0):=\mathbf0\), and let \(\mathbf P_{\mathcal X}\) denote the orthogonal projector onto \(\mathcal X\). Then
\begin{equation}
\label{eq:grad:U}
\nabla U(\mathbf x)
=
\mathbf P_{\mathcal X}
\sum_{\ell=1}^{s}
\mathbf M_\ell^\top\boldsymbol\phi(\mathbf M_\ell\mathbf x).
\end{equation}
For \(\mathbf z\neq\mathbf0\),
\[
J[\boldsymbol\phi](\mathbf z)
=
\|\mathbf z\|^{r-1}\mathbf I
+
(r-1)\|\mathbf z\|^{r-3}\mathbf z\mathbf z^\top
\succeq
r\|\mathbf z\|^{r-1}\mathbf I.
\]
Let \(\|\mathbf x\|,\|\mathbf y\|\leq R\). For each \(\ell=1,\ldots,s\), consider the segment \(\boldsymbol\xi_\ell(t):=\mathbf M_\ell\mathbf y+t\mathbf M_\ell(\mathbf x-\mathbf y)\), \(t\in[0,1]\). Since \(\|\boldsymbol\xi_\ell(t)\|\leq\|\mathbf M\|R\), absolute continuity gives
\[
\begin{aligned}
&(\mathbf M_\ell(\mathbf x-\mathbf y))^\top
\left(
\boldsymbol\phi(\mathbf M_\ell\mathbf x)
-
\boldsymbol\phi(\mathbf M_\ell\mathbf y)
\right)\\
&=
\int_0^1
(\mathbf M_\ell(\mathbf x-\mathbf y))^\top
J[\boldsymbol\phi](\boldsymbol\xi_\ell(t))
\mathbf M_\ell(\mathbf x-\mathbf y)\,\mathrm dt\\
&\geq
r(\|\mathbf M\|R)^{r-1}
\|\mathbf M_\ell(\mathbf x-\mathbf y)\|^2.
\end{aligned}
\]
The identity remains valid if the segment passes through the origin, since \(r>0\) and \(\boldsymbol\phi(\boldsymbol\xi_\ell(t))\) is absolutely continuous. Since \(\mathbf x-\mathbf y\in\mathcal X\), the projector in \eqref{eq:grad:U} disappears when taking the inner product with \(\mathbf x-\mathbf y\). Summing over \(\ell=1,\ldots,s\) therefore gives
\[
\begin{aligned}
&(\mathbf x-\mathbf y)^\top
\left(\nabla U(\mathbf x)-\nabla U(\mathbf y)\right)\\
&\geq
r(\|\mathbf M\|R)^{r-1}
\|\mathbf M(\mathbf x-\mathbf y)\|^2\geq
r(\|\mathbf M\|R)^{r-1}
\sigma_{\mathbf M}^2
\|\mathbf x-\mathbf y\|^2.
\end{aligned}
\]
Thus, \(U\) is strictly convex. Moreover, \(\mathbf x^\top\mathbf y-U(\mathbf x)\leq\|\mathbf x\|\|\mathbf y\|-c\|\mathbf x\|^{1+r}\to-\infty\) as \(\|\mathbf x\|\to\infty\), so \(U^*\) is finite everywhere. Hence, \cite[Theorem~26.5]{rockafellar1970} implies that \(\nabla U:\mathcal X\to\mathcal X\) is a bijection with continuous inverse \(\nabla U^*\), and \(U^*\) is differentiable.

Finally, fix \(M>0\). Since \(\nabla U^*\) is continuous, there exists \(R>0\) such that \(\|\nabla U^*(\mathbf y)\|\leq R\) whenever \(\|\mathbf y\|\leq M\). Thus, for \(\|\mathbf y\|,\|\widehat{\mathbf y}\|\leq M\), setting \(\mathbf x:=\nabla U^*(\mathbf y)\) and \(\widehat{\mathbf x}:=\nabla U^*(\widehat{\mathbf y})\), the preceding estimate and Cauchy--Schwarz give
\[
r(\|\mathbf M\|R)^{r-1}\sigma_{\mathbf M}^2
\|\mathbf x-\widehat{\mathbf x}\|
\leq
\|\mathbf y-\widehat{\mathbf y}\|.
\]
Therefore, \(\nabla U^*\) is locally Lipschitz.

\emph{Claim 2: For every \(\mathbf x\in\mathcal X\), the map \(\mathbf T_{\mathbf x}\) in \eqref{eq:power-potential-Tx} is well defined and linear.} The construction accounts for the fact that \(\nabla U\) need not be differentiable at points where some \(\mathbf M_\ell\mathbf x=\mathbf0\). Fix \(\mathbf x\in\mathcal X\), and define \(\mathcal Z(\mathbf x):=\{\ell:\mathbf M_\ell\mathbf x=\mathbf0\}\) and \(\mathcal K(\mathbf x):=\{\mathbf w\in\mathcal X:\mathbf M_\ell\mathbf w=\mathbf0,\ \ell\in\mathcal Z(\mathbf x)\}\). Thus, \(\mathcal Z(\mathbf x)\) contains the blocks at which \(\boldsymbol\phi\) is not differentiable, while \(\mathcal K(\mathbf x)\) contains the directions along which these blocks remain zero. For each \(\mathbf v\in\mathcal X\), define
\begin{equation}
\label{eq:power-potential-Tx}
\begin{aligned}
&\mathbf T_{\mathbf x}\mathbf v
:=\\&
\operatorname*{argmin}_{\mathbf w\in\mathcal K(\mathbf x)}
\left\{
\frac12
\sum_{\ell\notin\mathcal Z(\mathbf x)}
(\mathbf M_\ell\mathbf w)^\top
J[\boldsymbol\phi](\mathbf M_\ell\mathbf x)
\mathbf M_\ell\mathbf w
-
\mathbf v^\top\mathbf w
\right\}.
\end{aligned}
\end{equation}
The quadratic form is positive definite on \(\mathcal K(\mathbf x)\). Indeed, if it vanishes at \(\mathbf w\in\mathcal K(\mathbf x)\), then \(\mathbf M_\ell\mathbf w=\mathbf0\) for every \(\ell\notin\mathcal Z(\mathbf x)\), while membership in \(\mathcal K(\mathbf x)\) gives the same for \(\ell\in\mathcal Z(\mathbf x)\). Hence, \(\mathbf M\mathbf w=\mathbf0\), and injectivity of \(\mathbf M\) gives \(\mathbf w=\mathbf0\). The minimizer therefore exists and is unique, and its first-order optimality condition gives
\begin{equation}
\label{eq:power-potential-hessian-characterization}
\mathbf u^\top\mathbf v
=
\sum_{\ell\notin\mathcal Z(\mathbf x)}
\mathbf u^\top\mathbf M_\ell^\top
J[\boldsymbol\phi](\mathbf M_\ell\mathbf x)
\mathbf M_\ell\mathbf T_{\mathbf x}\mathbf v,
\qquad
\forall\,\mathbf u\in\mathcal K(\mathbf x).
\end{equation}
Uniqueness and linearity of this condition in \(\mathbf v\) imply that \(\mathbf T_{\mathbf x}\) is linear.

\emph{Claim 3: \(J[\nabla U^*](\mathbf y)=\mathbf T_{\mathbf x}\) for \(\mathbf y=\nabla U(\mathbf x)\).} To simplify the argument, in what follows we use the formal symbol \(o(\|\mathbf h\|)\) to denote a scalar or vector term, as appropriate, satisfying \(\lim_{\|\mathbf h\|\to0}\|o(\|\mathbf h\|)\|/\|\mathbf h\|=0\), and may denote a different term at each occurrence. For \(\mathbf h\in\mathcal X\), define \(\widehat{\mathbf x}:=\nabla U^*(\mathbf y+\mathbf h)\) and \(\Delta\mathbf x:=\widehat{\mathbf x}-\mathbf x\). Local Lipschitzness of \(\nabla U^*\) gives \(\|\Delta\mathbf x\|\leq C\|\mathbf h\|\) for all sufficiently small \(\mathbf h\) and some \(C>0\). Using \eqref{eq:grad:U}, taking the inner product of \(\mathbf h=\nabla U(\widehat{\mathbf x})-\nabla U(\mathbf x)\) with \(\Delta\mathbf x\) gives
\(
\mathbf h^\top\Delta\mathbf x
\geq
\sum_{\ell\in\mathcal Z(\mathbf x)}
\|\mathbf M_\ell\Delta\mathbf x\|^{1+r},
\)
where the terms with \(\ell\notin\mathcal Z(\mathbf x)\) are nonnegative by monotonicity of \(\boldsymbol\phi\), while for \(\ell\in\mathcal Z(\mathbf x)\) one has \(\mathbf M_\ell\mathbf x=\mathbf0\) and \(\mathbf M_\ell\widehat{\mathbf x}=\mathbf M_\ell\Delta\mathbf x\). Hence,
\begin{equation}
\label{eq:power-potential-zero-block-limit}
\sum_{\ell\in\mathcal Z(\mathbf x)}
\|\mathbf M_\ell\Delta\mathbf x\|^{1+r}
\leq
\mathbf h^\top\Delta\mathbf x
\leq
C\|\mathbf h\|^2.
\end{equation}
Since \(r<1\), for all \(\ell\in\mathcal Z(\mathbf x)\),
\(
\lim_{\|\mathbf h\|\to0}
\|\mathbf M_\ell\Delta\mathbf x\|/\|\mathbf h\|=0.
\)

Let \(\Delta\mathbf x_{\mathcal K}\) be the orthogonal projection of \(\Delta\mathbf x\) onto \(\mathcal K(\mathbf x)\). The maps \(\mathbf M_\ell\), \(\ell\in\mathcal Z(\mathbf x)\), have common kernel \(\mathcal K(\mathbf x)\), so their restriction to \(\mathcal K(\mathbf x)^\perp\) is jointly injective. By finite dimensionality and the preceding limits,
\(
\lim_{\|\mathbf h\|\to0}
\|\Delta\mathbf x-\Delta\mathbf x_{\mathcal K}\|/\|\mathbf h\|=0.
\)

For every \(\ell\notin\mathcal Z(\mathbf x)\), \(\boldsymbol\phi\) is differentiable at \(\mathbf M_\ell\mathbf x\). Since \(\|\Delta\mathbf x\|\leq C\|\mathbf h\|\), its first-order expansion gives
\[
\boldsymbol\phi(\mathbf M_\ell\mathbf x+\mathbf M_\ell\Delta\mathbf x)
-
\boldsymbol\phi(\mathbf M_\ell\mathbf x)
=
J[\boldsymbol\phi](\mathbf M_\ell\mathbf x)
\mathbf M_\ell\Delta\mathbf x
+
o(\|\mathbf h\|).
\]
Using the projection estimate and the finiteness of the number of blocks, for every \(\mathbf u\in\mathcal K(\mathbf x)\),
\[
\mathbf u^\top\mathbf h
=
\sum_{\ell\notin\mathcal Z(\mathbf x)}
\mathbf u^\top\mathbf M_\ell^\top
J[\boldsymbol\phi](\mathbf M_\ell\mathbf x)
\mathbf M_\ell\Delta\mathbf x_{\mathcal K}
+
o(\|\mathbf h\|)\|\mathbf u\|.
\]
Subtracting \eqref{eq:power-potential-hessian-characterization} with \(\mathbf v=\mathbf h\) and taking \(\mathbf u=\Delta\mathbf x_{\mathcal K}-\mathbf T_{\mathbf x}\mathbf h\) gives
\[
\begin{aligned}
&\sum_{\ell\notin\mathcal Z(\mathbf x)}
\bigl(\mathbf M_\ell(\Delta\mathbf x_{\mathcal K}-\mathbf T_{\mathbf x}\mathbf h)\bigr)^\top
J[\boldsymbol\phi](\mathbf M_\ell\mathbf x)
\mathbf M_\ell(\Delta\mathbf x_{\mathcal K}-\mathbf T_{\mathbf x}\mathbf h)\\
&\qquad=
o(\|\mathbf h\|)
\|\Delta\mathbf x_{\mathcal K}-\mathbf T_{\mathbf x}\mathbf h\|.
\end{aligned}
\]
Since the quadratic form is positive definite on \(\mathcal K(\mathbf x)\), its left-hand side is bounded below by \(c_{\mathbf x}\|\Delta\mathbf x_{\mathcal K}-\mathbf T_{\mathbf x}\mathbf h\|^2\) for some \(c_{\mathbf x}>0\). Hence, \(\|\Delta\mathbf x_{\mathcal K}-\mathbf T_{\mathbf x}\mathbf h\|=o(\|\mathbf h\|)\). Together with the projection estimate, this gives
\[
\lim_{\|\mathbf h\|\to0}
\frac{
\|\nabla U^*(\mathbf y+\mathbf h)-\nabla U^*(\mathbf y)-\mathbf T_{\mathbf x}\mathbf h\|
}{
\|\mathbf h\|
}
=
0.
\]
Thus, \(\nabla U^*\) is differentiable at \(\mathbf y\), with \(J[\nabla U^*](\mathbf y)=\mathbf T_{\mathbf x}\).

\emph{Claim 4: The map \(\mathbf x\mapsto\mathbf T_{\mathbf x}\) is continuous.} Let \(\{\mathbf x_j\}_{j=1}^{\infty}\subset\mathcal X\) satisfy \(\lim_{j\to\infty}\mathbf x_j=\mathbf x\), fix \(\mathbf v\in\mathcal X\), and define \(\mathbf w_j:=\mathbf T_{\mathbf x_j}\mathbf v\) for \(j=1,2,\ldots\). Since \(\nabla U(\mathbf x_j)\to\nabla U(\mathbf x)\), local Lipschitzness of \(\nabla U^*\) and Claim~3 give a uniform bound on \(\|\mathbf T_{\mathbf x_j}\|\) for all sufficiently large \(j\). Hence, \(\{\mathbf w_j\}_{j=1}^{\infty}\) is bounded.

We first show that the blocks vanishing at \(\mathbf x\) also vanish in the limit of \(\{\mathbf w_j\}_{j=1}^{\infty}\). Fix \(\ell\in\mathcal Z(\mathbf x)\). If \(\ell\in\mathcal Z(\mathbf x_j)\), then \(\mathbf M_\ell\mathbf w_j=\mathbf0\). Otherwise, taking \(\mathbf u=\mathbf w_j\) in \eqref{eq:power-potential-hessian-characterization} at \(\mathbf x_j\) and retaining only the \(\ell\)-th nonnegative term gives
\[
r\|\mathbf M_\ell\mathbf x_j\|^{r-1}
\|\mathbf M_\ell\mathbf w_j\|^2
\leq
\mathbf w_j^\top\mathbf v.
\]
Since \(\lim_{j\to\infty}\mathbf M_\ell\mathbf x_j=\mathbf0\), \(r-1<0\), and \(\{\mathbf w_j\}_{j=1}^{\infty}\) is bounded, it follows that \(\lim_{j\to\infty}\mathbf M_\ell\mathbf w_j=\mathbf0\). Thus, every accumulation point of \(\{\mathbf w_j\}_{j=1}^{\infty}\) belongs to \(\mathcal K(\mathbf x)\).

Consider any convergent subsequence \(\{\mathbf w_{j_k}\}_{k=1}^{\infty}\), where \(\{j_k\}_{k=1}^{\infty}\subset\mathbb N\) is strictly increasing, and let \(\lim_{k\to\infty}\mathbf w_{j_k}=\overline{\mathbf w}\). For all sufficiently large \(k\), every block with \(\ell\notin\mathcal Z(\mathbf x)\) also satisfies \(\ell\notin\mathcal Z(\mathbf x_{j_k})\). Moreover, every \(\mathbf u\in\mathcal K(\mathbf x)\) belongs to \(\mathcal K(\mathbf x_{j_k})\), while the terms with \(\ell\in\mathcal Z(\mathbf x)\setminus\mathcal Z(\mathbf x_{j_k})\) vanish because \(\mathbf M_\ell\mathbf u=\mathbf0\). Passing to the limit in \eqref{eq:power-potential-hessian-characterization} therefore gives
\[
\mathbf u^\top\mathbf v
=
\sum_{\ell\notin\mathcal Z(\mathbf x)}
\mathbf u^\top\mathbf M_\ell^\top
J[\boldsymbol\phi](\mathbf M_\ell\mathbf x)
\mathbf M_\ell\overline{\mathbf w},
\qquad
\forall\,\mathbf u\in\mathcal K(\mathbf x).
\]
Since \(\overline{\mathbf w}\in\mathcal K(\mathbf x)\), uniqueness in \eqref{eq:power-potential-hessian-characterization} gives \(\overline{\mathbf w}=\mathbf T_{\mathbf x}\mathbf v\). Thus, every accumulation point of \(\{\mathbf w_j\}_{j=1}^{\infty}\) equals \(\mathbf T_{\mathbf x}\mathbf v\), and therefore \(\lim_{j\to\infty}\mathbf T_{\mathbf x_j}\mathbf v=\mathbf T_{\mathbf x}\mathbf v\). Applying this argument to a fixed basis of \(\mathcal X\) gives continuity of \(\mathbf T_{\mathbf x}\) in matrix norm. By Claims 3 and 4, \(J[\nabla U^*](\mathbf y)=\mathbf T_{\nabla U^*(\mathbf y)}\) is continuous because \(\nabla U^*\) is continuous. Therefore, \(U^*\) is twice continuously differentiable.

\end{proof}

\subsection{Proof of the application corollaries}
\label{sec:proof:corollaries}

\begin{proof}[Proof of Corollaries~\ref{cor:leader}--\ref{cor:affine}]
We verify that the three error systems satisfy Definition~\ref{def:ahc} with the choices in Table~\ref{tab:applications}. For the leader--follower problem, \(\mathbf L+\mathbf B\) is symmetric positive definite because the graph is connected and at least one \(b_i\neq0\). For the leaderless problem, \(\mathbf D^\top\) is injective on \(\mathbf1^\perp\), while for affine formation tracking, \(\mathbf\Omega_{FF}\) is symmetric positive definite. Hence, the corresponding matrices \(\mathbf Q\) are symmetric positive definite, and Proposition~\ref{prop:power-potential}, applied with \(\mathbf M=\mathbf L+\mathbf B\), \(\mathbf M=\mathbf D^\top\), and \(\mathbf M=\mathbf\Omega_{FF}\), verifies all the required properties of \(U_\mu\) and gives the expressions for \(\mathbf Q\nabla U_\mu\) in Table~\ref{tab:applications}.

In all three cases, \(\mathcal S\) is upper semicontinuous, has nonempty compact convex values, and is homogeneous of degree zero. Moreover, for every \(\mathbf s\in\mathcal S(\mathbf q)\), the leader--follower case satisfies \(\mathbf q^\top\mathbf s=\|(\mathbf L+\mathbf B)\mathbf q\|_1\geq\lambda_{\min}(\mathbf L+\mathbf B)\|\mathbf q\|\), the leaderless case satisfies \(\mathbf q^\top\mathbf s=\|\mathbf D^\top\mathbf q\|_1\geq\sqrt{\lambda_2(\mathbf L)}\|\mathbf q\|\), and the affine case satisfies \(\mathbf q^\top\mathbf s=\sum_{i=1}^{N_F}\|[\mathbf\Omega_{FF}\mathbf q]_i\|\geq\lambda_{\min}(\mathbf\Omega_{FF})\|\mathbf q\|\). Thus, the values of \(c_{\mathcal S}\) in Table~\ref{tab:applications} satisfy Definition~\ref{def:ahc}.

For Corollary~\ref{cor:leader}, define \(\mathcal D:=\{-(a/L)\mathbf B\mathbf1:|a|\leq L/\sqrt N\}\). This set is nonempty, compact, and convex, and \(\|\mathbf d\|\leq1\) for every \(\mathbf d\in\mathcal D\), since \(\|\mathbf B\mathbf1\|\leq\sqrt N\). Since \((\mathbf L+\mathbf B)\mathbf1=\mathbf B\mathbf1\), Assumption~\ref{ass:leader} gives \(-z_0^{(m+1)}(t)\mathbf B\mathbf1/L\in\mathcal D\) and \(L(\mathbf L+\mathbf B)^{-1}[-z_0^{(m+1)}(t)\mathbf B\mathbf1/L]=-\mathbf1z_0^{(m+1)}(t)\). Hence, \eqref{eq:leader:error} is an AHC with \(\gamma=0\). Theorem~\ref{thm:main} gives gains for which its origin is globally finite-time stable, and \eqref{eq:leader:error:coordinates} then gives \eqref{eq:leader:objective}.

For Corollary~\ref{cor:redcho}, let \(\mathcal D:=\{\mathbf d\in\mathbf1^\perp:\|\mathbf d\|\leq1\}\). By the definition of \(\mathbf v(t)\), \(\mathbf1^\top\mathbf v(t)=0\), while Assumption~\ref{ass:redcho} gives \(\|\mathbf v(t)\|\leq L\). Hence, \(\mathbf d(t)=\mathbf v(t)/L\in\mathcal D\), and \eqref{eq:redcho:error} is the abstract homogeneous chain specified in Table~\ref{tab:applications} on \(\mathcal X=\mathbf1^\perp\). Theorem~\ref{thm:main} guarantees gains \(k_0,\ldots,k_m>0\) for which its origin is globally finite-time stable. For \(\gamma>0\), the implication established in \cite{redcho} then yields \eqref{eq:distro:diff}, while for \(\gamma=0\), under \eqref{eq:initial}, the corresponding implication follows from \cite{edcho}.

For Corollary~\ref{cor:affine}, let \(\mathcal D:=\{\mathbf d\in\R^{dN_F}:\|\mathbf d\|\leq1\}\). Assumption~\ref{ass:affine} gives \(\mathbf\Omega_{LF}^{\top}\mathbf p_L^{(m+1)}(t)/L\in\mathcal D\), and multiplication by \(L\mathbf Q=L\mathbf\Omega_{FF}^{-1}\) gives the disturbance in \eqref{eq:affine:error}. Thus, \eqref{eq:affine:error} is an AHC with \(\gamma=0\). Theorem~\ref{thm:main} gives gains for which its origin is globally finite-time stable, and \eqref{eq:affine:error:coordinates} then gives \eqref{eq:affine:objective}.
\end{proof}



\end{document}